\documentclass[compsoc, conference, letterpaper, 10pt, times]{IEEEtran}
\ifCLASSOPTIONcompsoc
  \usepackage[nocompress]{cite}
\else
  \usepackage{cite}
\fi
\ifCLASSINFOpdf
   \usepackage[pdftex]{graphicx}
\else
\fi
\usepackage{amsmath}
\usepackage{amsthm}
\usepackage{amsfonts}
\usepackage{amssymb}
\usepackage[nolist,nohyperlinks]{acronym}
\usepackage{framed}
\newcommand*{\estimatesB}{\mathrel{\widehat=}}
\newtheorem{definition}{Definition}
\newtheorem*{remark}{Remark}
\newtheorem{theorem}{Theorem}
\usepackage[table]{xcolor}
\usepackage[many]{tcolorbox}
\usepackage{enumerate}
\usepackage{multirow}               
\usepackage{array}                  
\usepackage{tabularx}
\usepackage{longtable}
\usepackage{algorithmic}
\usepackage{url}

\usepackage{tikz}
\usetikzlibrary{shapes.geometric}
\usetikzlibrary{arrows.meta,arrows}
\usepackage[utf8]{inputenc}
\usepackage{pgfplots}
\usepackage{scalefnt}

\usepackage{colonequals}
\newcommand*{\logeq}{\ratio\Leftrightarrow}

\definecolor{mygreen}{RGB}{117,167,117}
\definecolor{myred}{RGB}{255,1,1}

\newcommand{\FUSION}{\widehat{\oplus}}
\newcommand{\FUSIONC}{\widehat{\oplus}_c}

\newcommand{\bitem}{\begin{itemize}}
\newcommand{\eitem}{\end{itemize}}

\newenvironment{boxfig}[3]{%
     \begin{figure}[ht]
     \newcommand{\FigCaption}{#2}
     \newcommand{\FigLabel}{#3}
     
     \begin{shadebox}[adjusted title=\emph{\textbf{#1}}]

\begin{center}
\begin{small}
\begin{minipage}{0.9\linewidth}
\medskip
}{
\smallskip
\end{minipage}
\end{small}
\end{center}
\end{shadebox}
\smallskip
  \caption{\FigCaption}
    \label{\FigLabel}
\end{figure}
}

\newtcolorbox{shadebox}[1][]{
  enhanced, 
  boxrule=1.2pt,arc=1pt,boxsep=-1mm,
  left=2pt,right=2pt,top=5pt,bottom=5pt,width=0.48\textwidth, colbacktitle=black!40!white,enhanced,
attach boxed title to top center={yshift=-3pt},
  colback=white,#1
}

\begin{document}
%
\title{M-STAR: A Modular, Evidence-based Software Trustworthiness Framework}



%
\author{\IEEEauthorblockN{Nikolaos Alexopoulos\IEEEauthorrefmark{1},
Sheikh Mahbub Habib\IEEEauthorrefmark{1},
Steffen Schulz\IEEEauthorrefmark{2}and
Max M{\"u}hlh{\"a}user\IEEEauthorrefmark{1}
\medskip
\IEEEauthorblockA{\IEEEauthorrefmark{1}
Technische Universit{\"a}t Darmstadt, Germany\\
\{alexopoulos, sheikh, max\}@tk.tu-darmstadt.de}}
\smallskip
\IEEEauthorblockA{\IEEEauthorrefmark{2}Intel Labs\\
steffen.schulz@intel.com}}


\maketitle

\begin{abstract}
Despite years of intensive research in the field of software vulnerabilities discovery,
exploits are becoming ever more common.
Consequently, it is more necessary than ever to choose software configurations that
minimize systems' exposure surface to these threats.
In order to support users in assessing
the security risks induced by their software configurations and in making informed
decisions, we introduce M-STAR, a
Modular Software Trustworthiness ARchitecture and framework for
probabilistically assessing the trustworthiness of software systems,
based on evidence, such as their vulnerability history and source code properties.

Integral to M-STAR is a software trustworthiness model, consistent with
the concept of computational trust.
Computational trust models
are rooted in Bayesian probability and Dempster-Shafer
Belief theory,
offering mathematical soundness and expressiveness to our framework.
To evaluate our framework, we instantiate M-STAR for
Debian Linux packages, and investigate real-world deployment scenarios.
In our experiments with real-world data, M-STAR could assess the relative trustworthiness
of complete software configurations with an error of less than 10\%.
Due to its modular design, our proposed framework is agile, as it can incorporate
future advances in the field of code analysis and vulnerability prediction.
Our results point out that M-STAR can be a valuable tool for system administrators,
regular users and developers, helping them assess and manage risks associated with
their software configurations.
\end{abstract}


%
\IEEEpeerreviewmaketitle

%
\section{Introduction}
Modern software systems comprise a multitude of interacting components
developed by different developers.
The security of such systems is of foremost importance, as they are used in various
critical aspects of our everyday lives, such as telecommunications, hospitals,
transportations, etc. The recent (May 2017) ``WannaCry'' exploit~\cite{wiki:wannacry} showed the
scale of disruption even a known and patched vulnerability can incur. This exploit was
made possible because of a vulnerability in Microsoft's implementation of the SMB server
that allowed remote attackers to execute arbitrary code on the victim's machine. The
vulnerability was disclosed by Microsoft in March 2017 (CVE-2017-0144) but the machines
that were left unpatched were the targets of the attack that cost an estimated \$8 Billion
in system downtime, according to a recent report~\cite{hackerone2017report}. The fact that known
and patched vulnerabilities can cause such great disturbance is an indicator that yet unknown vulnerabilities (zero-days),
which can potentially affect billions of devices, are a great danger.
Zero-day exploits are a major attack vector and over 20 thousand new vulnerabilities
were discovered through HackerOne's bug bounty program in 2016 alone, according to the same report, while
the amount of CVEs reported in 2017 were more than double compared to any previous year.

The threat of software exploits is therefore at an all-time high, even though
the security community has come up with various defense mechanisms to locate and fix vulnerabilities,
including \textit{formal verification, static/dynamic analysis of code} and \textit{fuzzing}.
Formal verification of software is a way to achieve very strong security guarantees, effectively
rendering vulnerabilities impossible. However, formal verification incurs high overhead, requires
manual labor, and is meant to be applied on inherently high-risk components, such as the ones implementing
cryptographic protocols. Interestingly, in reality, even cryptographic components such as openssl 
are not formally verified, as they include numerous optimizations.
Consequently, in recent years, the research community has produced
several advances on automated vulnerability discovery.
State of the art static analysis tools, like the ones described in~\cite{brown2017finding, biswas2017venerable, wang2017how}
offer ways to check software for possible vulnerabilities pre-release by pinpointing risky code segments.
Additionally, there has been no lack of progress in the field of dynamic analysis and fuzzing tools, e.g.~\cite{pan2017digtool, schumilo2017kafl, stephens2017driller},
that discover vulnerabilities via testing the runtime behaviour of the program.
Even after applying various proactive pre-release measures, as mentioned above, large end-products most often contain vulnerabilities. This is evident by the high rate of security patching found in almost all big software products. These vulnerabilities can in turn lead to major security exploits.

Hence, it is necessary to quantify associated risks and assess the trustworthiness of software components,
so as, among others, (i) users and system administrators can make decisions regarding which software components to install,
(ii) companies can assess and attest the trustworthiness of employee devices having access to sensitive data and
(iii) developers can make decisions regarding which components to use as dependencies for their software.
Towards this goal, we propose a novel software trustworthiness framework that can process evidence regarding the security
history of isolated software components and complex software systems.
Our framework is designed in a modular fashion in order to adapt to the requirements of the user, can take into
account additional available evidence (e.g. static analysis results), and can readily incorporate future advances
in prediction mechanisms.
Central to our approach are prediction techniques and computational trust
models, together with operators that enhance these models in order to handle
system-wide trustworthiness assessments.
\\
\\\noindent\textbf{Our Contributions:}
The contributions of this paper can be summarized by the following points:
\begin{itemize}
	\item An intuitive and mathematically grounded trust model for software.
	\item A deep learning based vulnerability prediction technique, harnessing
		historical evidence.
	\item A modular architecture for secure software trustworthiness assessment,
		incorporating our model and prediction results.
	\item A detailed analysis of the vulnerability landscape of Debian.
		trustworthiness.
	\item An application of our framework on real-world systems.
\end{itemize}
In this paper, we present a novel modular architecture for assessing software trustworthiness based on the security history of software. We also implement a proof-of-concept system for assessing the trustworthiness of systems consisting of Debian Linux packages. As part of our deep investigation of the Debian vulnerability landscape, we came to the conclusion that the
current practice of vulnerability discovery is like scratching  off the tip of an iceberg; it rises up a little, but we (the security community) are not making any visible progress.
Motivated by this result, we propose an approach to calculate the trustworthiness of a software component (e.g. the Linux kernel, openssl, Firefox, etc.) w.r.t. the vulnerabilities predicted to affect the
specific component in the future. In order to predict vulnerabilities, we consider past security evidence that are mined from publicly accessible vulnerability repositories, such as
the National Vulnerability Database (NVD) and Debian Security Advisories (DSA). We then proceed
to compare different predictive mechanisms for future vulnerabilities w.r.t. their accuracy.
Our deep learning technique, employing LSTMs achieves almost 20\% better performance than
the other proposed heuristic methods.
The resulting predictions are fed to a novel, formal probabilistic trust model that
takes into account the projected (un)certainty of the predictions. Then, we show how to combine the
component-wise trustworthiness assessments to system-wide trust scores, and communicate this scores using
an intuitive trust visualization framework to the user.

The long-term utility of our approach stems from its modular architecture. Our trust model can accommodate
predictions coming from various different prediction techniques as input.
For example, recent advances in applying machine
learning for software security, show great promise and future prediction methodologies can straightforwardly
be incorporated in our system.
\\
\\\noindent\textbf{Paper Organization:}\\
After going through some necessary background knowledge in section~\ref{sec:back}, we present
a high level overview of our system in section~\ref{sec:arch}. Then, we study the security
ecosystem of Debian in section~\ref{sec:vuln}, and proceed with our prediction analysis in
section~\ref{sec:pred}. Next, we present our software trust model in section~\ref{sec:model},
before we apply our framework to real-world system configurations~\ref{sec:apl}. Finally, we go over
the related work in section~\ref{sec:related} and conclude in section~\ref{sec:concl}.

%
\section{Background}\label{sec:back}
In this section we briefly go over some necessary material for the
comprehension of the paper.
\subsection{Bugs, vulnerabilities and exploits}
Real-world security incidents, i.e. \textit{exploits}, are attributed to flaws
in the source code of software products. Generally,
flaws in the source code of a program are referred to as \textit{bugs}. The
subset of bugs that can lead to security exploits are distinguished
as \textit{vulnerabilities}. We do not make a distinction between accidentally
created vulnerabilities (regular bugs) and maliciously placed ones (back-doors).
Exploits can take advantage either of publicly
known, yet unpatched vulnerabilities, or of yet unknown vulnerabilities. The
latter are known as \textit{zero-day vulnerabilities}. Protecting computer
systems against known vulnerabilities comes down to effectively applying
patches to the systems, while protection against zero-day vulnerabilities
is more difficult, relying on the correct choice of software, in order to
minimize the inherent risk.
\subsection{The Debian Linux distribution}
Debian GNU/Linux is a particular distribution of the Linux
operating system, and numerous packages that run on it \cite{Debian}.
There are over $40,000$ software packages available through the Debian
distribution at present and users can choose which of them to install on
their system. All packages included in Debian are open source and free to
redistribute, usually under the terms of the GNU General Public License
\cite{stallman1991gnu}.

Security incidents, i.e. vulnerabilities, are handled in a transparent
manner by the Debian security team \cite{debian-dsa}. The security team
reviews incident notifications for the stable release and after working
on the related patches, publishes a \ac{DSA}.

\subsection{Predictive analytics}
Predicting future events is of paramount importance to trust and risk assessment
methodologies. In our context, the events in question are security incidents
(vulnerabilities) affecting software components, and more specifically for
the case of our study, Debian packages. Consequently, our problem can be
viewed as a time-series prediction problem, where we want to predict the
number of vulnerabilities of a software component in the future by taking
advantage of its vulnerability history, and optionally, some other related
information that is available (e.g. stemming from static analysis).

There exist a variety of forecasting techniques for time-series data. These
can vary from basic simple predictors based on universal observations, e.g.
average or weighted average of the observations, to linear autoregressive
models, e.g. so called ARIMA models. Additionally, recent advances in using
machine learning techniques on various predictive tasks, including in
software security, indicate that supervised learning models, such as
\acp{SVM} often provide good predictions, although they have been mainly
applied in classification problems. \ac{LSTM} neural networks have gained
momentum in the last couple of years in tasks related to the forecasting
of time-dependent processes in various scientific and industrial fields.
\acp{LSTM} are recurrent artificial neural networks whose units are
designed to remember values for long or short periods, making them
especially suitable for the prediction of the next steps
in a time-series. In this paper, we employ \acp{LSTM} and evaluate their
effectiveness in the context of vulnerability prediction based on a software
component's history.
\subsection{Computational trust}\label{sec:comptrust}
Computational trust provides means to support entities make informed decisions
in electronic environments, where decisions are often subject to risk and uncertainty.
Research in computational trust addresses the formal modeling, assessment, and
management of the social notion of trust for use in risky electronic environments.
Examples of such risky environments span from social networks to
cloud computing service ecosystems~\cite{Habib2015BookChap}. In theory, trust is usually reasoned in terms of
the relationship within a specific context between a \emph{trustor} and a
\emph{trustee}, where the trustor is a subject that trusts a target entity, which is
referred to as trustee. Mathematically, trust is an estimate by the trustor of the
inherent quality of the trustee, i.e., the quality of another party to act
beneficially or at least non-detrimentally to the relying party. This estimate is
based on evidence about the trustee's behaviour in the past, in the case of
M-STAR past vulnerabilities and characteristics of software.
%

CertainTrust~\cite{Ries2009Extending} (and its accompanying algebra of operators,
CertainLogic~\cite{ries2011certainlogic}) or Subjective Logic~\cite{josang2001logic} provide
mathematically grounded models (consistent with Bayesian statistics) to represent and compute trustworthiness under uncertain probabilities.
In this paper, we use CertainTrust and CertainLogic for trust representation and computation
respectively, although Subjective Logic can also be used. In CertainTrust, trustworthiness is represented
by a  construct called opinion ($o$). Opinions express the truth of a statement or proposition.
The opinion is a tuple, i.e. $o= (t,c,f)$. Here, the \emph{trust value} $t$ indicates the most likely
value for the estimated parameter (in Bayesian terms, the mode of the posterior distribution).
It can depend on the relative frequency of observations or pieces of evidence supporting the truth of a proposition.
The \emph{certainty value} $c$ indicates the degree to which the $t$ is assumed to be representative
for the future (associated to the credible interval). The higher the certainty ($c$) of an opinion is,
the higher the influence of the trust value
($t$), on the expectation value $E$ (trustworthiness score), in relation to the initial expectation value
($f$). The parameter $f$ expresses the assumption about the truth of a proposition in absence of evidence (prior distribution).
The expectation value (trustworthiness score), $E$, is calculated as follows:
$E= t \cdot c + (1-c) \cdot f$. $E$ expresses an estimation of trustworthiness considering $t$, $c$, and $f$.

CertainLogic, based on the CertainTrust model, provides mathematical operators to aggregate multiple
opinions (supporting the truth of propositions) considering uncertainty and conflict. It offers a
set of standard operators like $AND^{ct}$ ($\wedge^{ct}$), $OR^{ct}$ ($\vee^{ct}$),
and $NOT^{ct}$ ($\neg^{ct}$) as well as non-standard operators like Consensus ($\oplus$), Discounting ($\otimes$),
and Fusion ($\FUSION$). The standard operators are defined to combine opinions associated with
propositions that are independent. The non-standard operators are defined to combine opinions associated
with propositions that can also be dependent.
\section{System Architecture}\label{sec:arch}
Our main goal is to develop a modular architecture that
is intuitive and easily extensible. Thus, we identify five
core components of our system, as shown in Figure \ref{fig:arch}.
We now proceed to the specification of the components in a
high level and in the next sections we elaborate on the approaches and solutions used in this paper.\\
\begin{figure}[h]
	\centering
	\caption{System Architecture}
	\label{fig:arch}
	\includegraphics[scale=0.4]{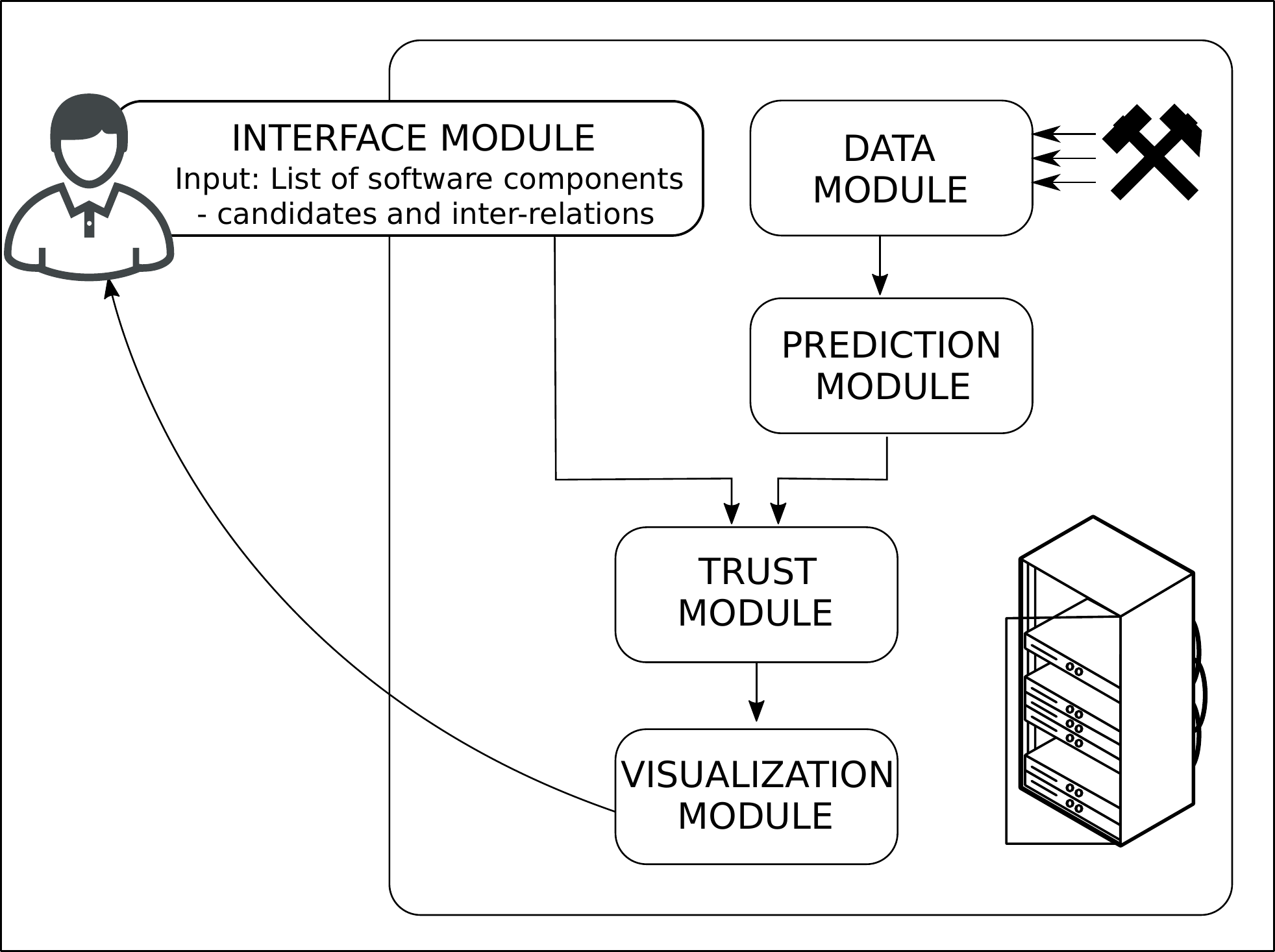}
\end{figure}
\\\smallskip
\noindent\textbf{Interface Module:} \\
The Interface module provides information about which components comprise
the system under evaluation, along with information about their inter-dependencies.
As part of the Interface module, an attestation protocol (e.g.~\cite{sailer2004design}) can be implemented
to guarantee that the data provided as input has not been tampered with.\\
\\\smallskip
\noindent\textbf{Data Module:} \\
The Data module consists of the evidence-gathering mechanisms
employed by the framework implementation. In our instantiation of M-STAR, we mine the Debian Security
Advisories and NVD's CVE reports for past vulnerabilities of the
software components.
\\
\\\smallskip
\noindent\textbf{Prediction Module:} \\
The Prediction module includes the data analysis mechanism used
to predict future vulnerabilities of the components. This
mechanism can range from simple averaging models to complex
machine learning approaches.
\\
\\\smallskip
\noindent\textbf{Trust Module:} \\
The Trust module is responsible for modeling and calculating the trustworthiness
of each component. The module also combines individual assessments
to system-wide trustworthiness values, depending on the configuration of the
target system.
\\
\\\smallskip
\noindent\textbf{Visualization Module:} \\
The Visualization module serves the important purpose of
communicating the resulting trustworthiness scores to the user or system administrator. The module provides intuitive and comprehensible graphical as well as numerical interfaces in order to assist their decision making.
\\

\section{Vulnerabilities in large software projects}\label{sec:vuln}
Our framework requires good-quality (i.e. correct and complete) sources of data regarding past vulnerabilities
of software components. Therefore, a well-organized and maintained security
report repository is required. Such vulnerability repositories are
maintained by entities such as large companies (e.g. big software vendors like Microsoft, or anti-virus companies),
intelligence agencies,
or big open-source projects. For our deployment, we choose the Debian distribution
of GNU/Linux as the focus of our efforts, based on the comprehensive variety of software offered as part
of it, and its transparent, open, and security-focused order of operation.
\subsection{Vulnerabilities in Debian}In this section, we present an overview of the Debian
ecosystem w.r.t. its security characteristics and draw some interesting conclusions.
We collected our data (implementation of the Data Module) from the \acp{DSA}
and NIST's NVD\footnote{https://nvd.nist.gov/}.
\begin{figure*}[h]
\centering
\includegraphics{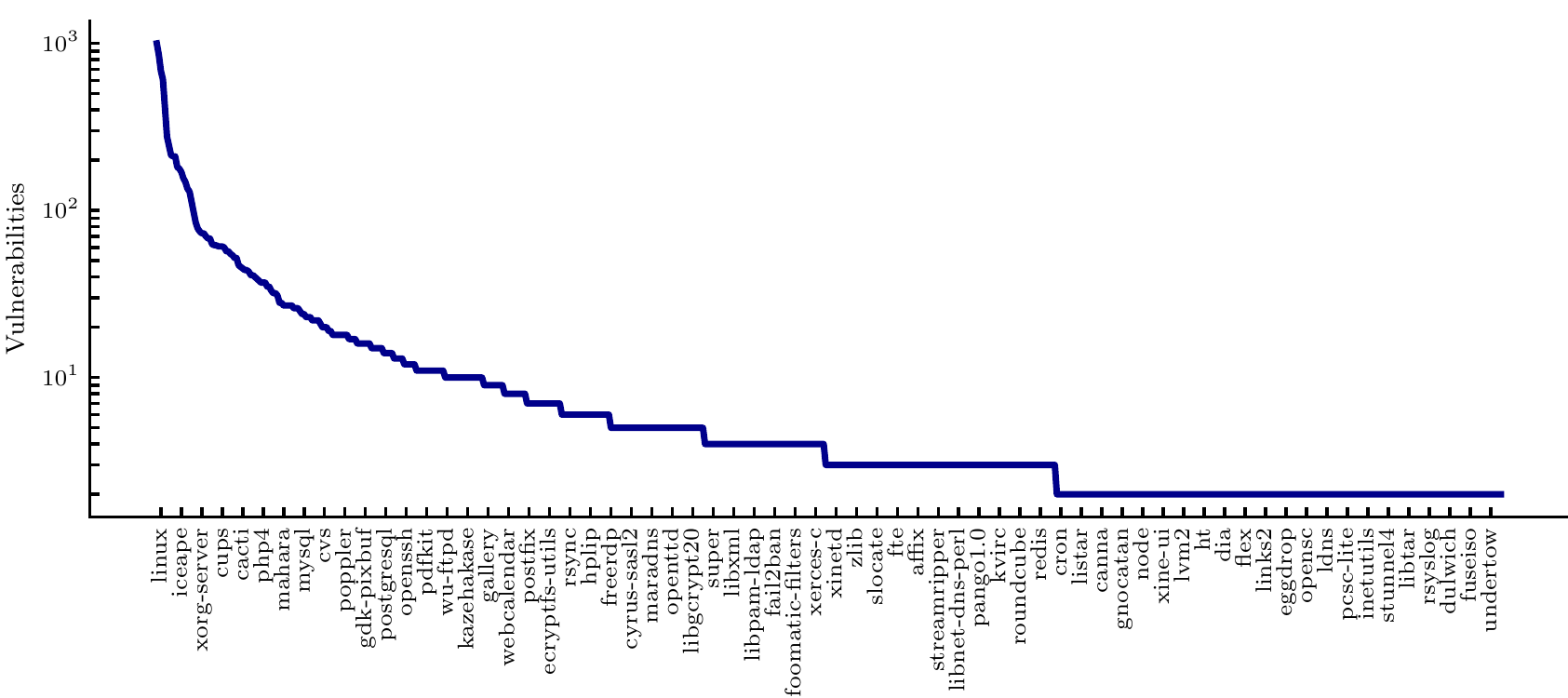}
\caption{The distribution of vulnerabilities in the Debian ecosystem (years 2001-2016). The scale of axis y is logarithmic. Only packages with at least two vulnerabilities are taken into account. Every tenth package name appears on the x axis for space reasons.}
\label{fig:distr}
\end{figure*}
In Fig. \ref{fig:distr}, we see the distribution of vulnerabilities among source packages of Debian for the years 2001-2016.
In the figure, packages that had at least two vulnerabilities in the specified time frame are included, yielding a total of 619 source packages.
An additional 540 source packages had a single vulnerability and were not included in the figure for readability reasons (the complete figure is
available in the appendix).
It is interesting to notice that the distribution is characteristically heavy-tailed (notice that the y axis is logarithmic)
with a few source packages dominating the total vulnerabilities reported and a long tail of a large number of packages with
only a few vulnerabilities. Interestingly, inspecting the plot (Fig.~\ref{fig:loglog}) of the data in (double) logarithmic axes, we
can observe a straight line, indicative of a power-law distribution.
In short, we observe that the majority of vulnerabilities is concentrated in a small set of source packages,
and therefore the trustworthiness of software systems largely depends on which of those high-risk packages those systems use.

Table \ref{tab:most} presents the 20 top vulnerable Debian source packages of all time. An automated procedure was established
to collect the vulnerabilities reported for previous versions of a package and attribute them to the current
version of the package in the stable distribution. Manual checks and small corrections were subsequently performed.
The Linux kernel, as probably expected, is the most vulnerable component,
followed by the two main browsers in use (Firefox\footnote{Firefox ESR (Extended Support Release) is the version of Mozilla Firefox packaged
in Debian} and Chromium\footnote{Chromium consists of the open-source code-base of the
proprietary Google Chrome browser (https://www.chromium.org/)}). The total number of vulnerabilities reported in the 16 year period was
10\,747, meaning the Linux kernel accounts for more than 10 percent of the total. During the previous two years (2015-2016), a total of
2\,339 vulnerabilities were reported, with Chromium being by far the most affected package, accounting for 297 vulnerabilities compared to
the next most vulnerable package (Firefox) which was affected by 153. The Linux kernel, in that time period was affected by 144 vulnerabilities,
which is roughly 6 percent of the total.
\begin{table}[htbp]
\caption{The top twenty packages with the most vulnerabilities in time periods (i) 2001-2016 and (ii) 2015-2016, along with the respective ranks of the source packages w.r.t. their vulnerabilities.}
	\centering
\begin{tabular}{|l|c|c|c|c|}
\hline
Source package name & \# total & rank total & \# 15-16 & rank 15-16 \\ \hline \hline
linux & 1303 & 1 & 144 & 3 \\ \hline
firefox-esr & 815 & 2 & 153 & 2 \\ \hline
chromium-browser & 496 & 3 & 297 & 1 \\ \hline
openjdk-8 & 353 & 4 & 121 & 4\\ \hline
icedove & 347 & 5 & 89 & 5\\ \hline
wireshark & 261 & 6 & 87 & 6\\ \hline
php7.0 & 258 & 7 & 86 & 7\\ \hline
mysql-transitional & 221 & 8 & 63 & 10\\ \hline
xulrunner & 211 & 9 & -- & --\\ \hline
iceape & 178 & 10 & -- & --\\ \hline
openssl & 145 & 11 & 50 & 13\\ \hline
qemu & 134 & 12 & 70 & 8\\\hline
xen & 113 & 13 & 52 & 12 \\ \hline
wordpress & 110 & 14 & 38 & 15 \\ \hline
tomcat8 & 99 & 15 & 48 & 14\\ \hline
imagemagick & 95 & 16 & 57 & 11 \\ \hline
krb5 & 89 & 17 & 10 & 39 \\\hline
typo3-src & 77 & 18 & 1 & 151-253 \\ \hline
ruby2.3 & 75 & 19 & 5 & 56-69\\ \hline
postgresql-9.6 & 75 & 20 & 19 & 22\\ \hline
\end{tabular}
\label{tab:most}
\end{table}
\begin{figure}[h]
\centering
\includegraphics{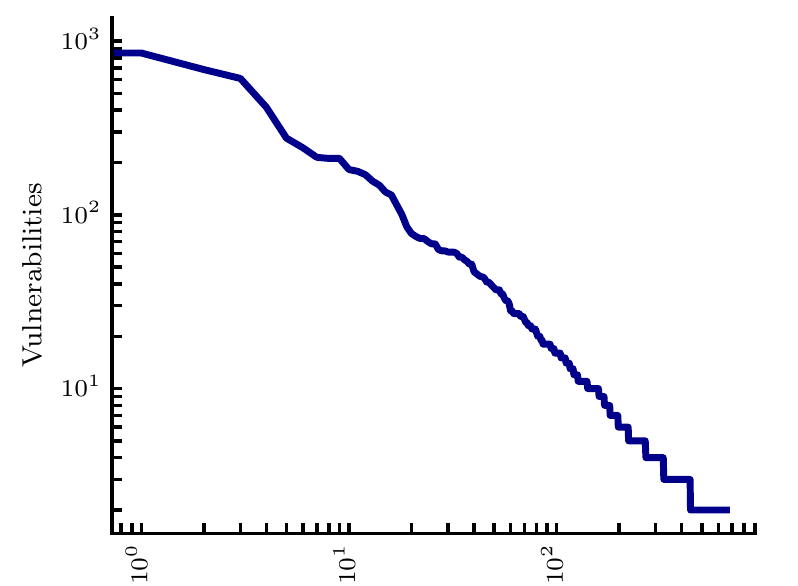}
\caption{A log-log plot of the distribution of Fig. \ref{fig:distr}.}
\label{fig:loglog}
\end{figure}

Concerning the total number of vulnerabilities reported in the Debian ecosystem w.r.t. time, Fig. \ref{fig:year}
shows a clear upward trend as the years go by.
\begin{figure}[h]
\centering
\includegraphics{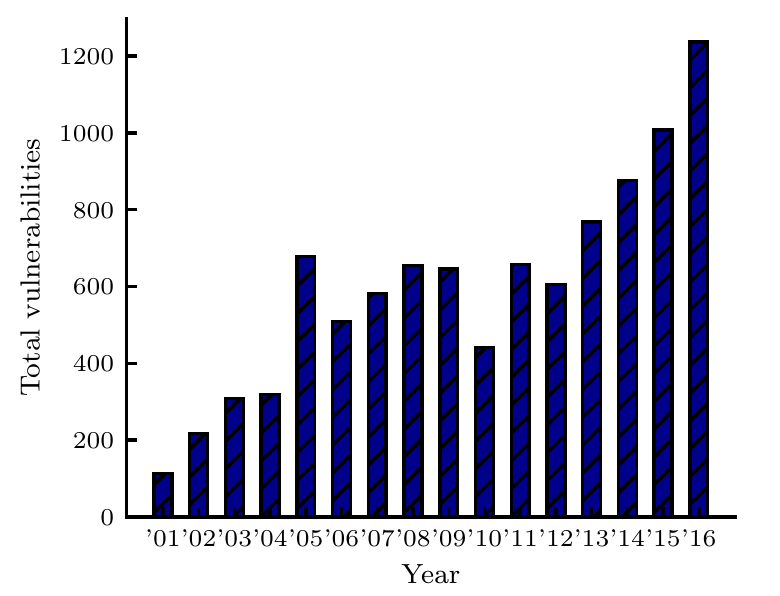}
\caption{Total vulnerabilities reported in Debian per calendar year (2001-2016).}
\label{fig:year}
\end{figure}
Can this mean that the security quality of the software is decreasing, despite
the considerable effort of security researchers and professionals? One could argue, that the amount of software packages of Debian
increased dramatically in recent years and this is the cause of the increase in the total amount of vulnerabilities
reported. The line of thought would be, that with such a large number of packages, even one or two bugs that slipped
the security measures of the individual maintainers, would contribute to a big yearly sum.
That would be a reasonable assumption, as the stable version of Debian released in 2002 (Woody) contained only 8\,500 binary packages,
going up to 18\,000 packages with the release of Etch in 2007, significantly increasing to 36\,000 in 2013 (Wheezy) and peaking at 52\,000 with the current stable version
(Stretch) released in June 2017. However, we found evidence supporting the opposite. Interestingly, the number of vulnerabilities per package (among the packages
that had a vulnerability reported for the specified year) also follows an upward trend, a fact straightforwardly deduced from Fig.~\ref{fig:average_per_year}.
\begin{figure}[h]
\centering
\includegraphics{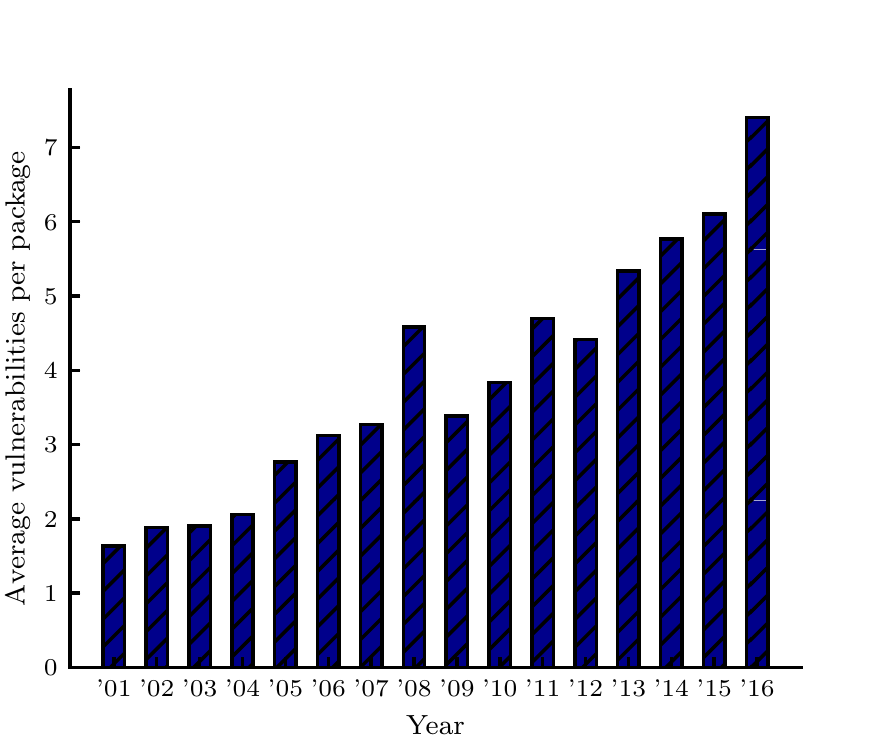}
\caption{Average vulnerabilities per package (that had at least one security incident in that year) in Debian per calendar year (2001-2016).}
\label{fig:average_per_year}
\end{figure}
In the latter figure we can even see a smoother, clearer upward trend compared to Fig.~\ref{fig:year}, although the slope of the trend is nearly identical.
These observations, together with our previous assessment that the distribution of vulnerabilities among the packages of Debian can be
attributed to a power law generation mechanism, indicate that there are specific packages that continue to have large numbers of vulnerabilities
for prolonged periods of time. Is there an explanation for this phenomenon or are we (security researchers) doing such a bad job? One possible
glimmer of hope would be if vulnerabilities were induced by software upgrades and the number of vulnerabilities affecting a specific version
of a package gradually dropped to zero.
An intuitive hypothesis would be that at least for certain stable versions of a package, the rate of
vulnerabilities will eventually decrease. In order to test the claim that specific
versions of a package reach a relatively secure state (few vulnerabilities reported per quarter) and that subsequent vulnerabilities that are
attributed to the package are caused by updates, we perform a case study on two popular packages, namely PHP and OpenJDK, which recently underwent
major version changes.
The hypothesis is that each major version of a package becomes more secure as time passes, as a result of the hard work of the security community
and that a considerable amount of new vulnerabilities affect only the new code inserted with the major updates. To test this hypothesis we inspected the vulnerabilities
reported for the newer versions of the packages and checked if
they also affect older versions.
\medskip
\par\noindent\textbf{PHP:}\smallskip\\
PHP is a popular server-side scripting language that is used by 83\% of all websites whose server-side programming language is known,
according to W3Techs \footnote{https://w3techs.com/technologies/details/pl-php/all/all} (measured in November 2017).
Several PHP versions have been packaged in Debian, traditionally following the source package name phpx, where x is the version number.
We will look into the transition from php5 to the next version php7 \footnote{Official package name php7.0} (php6 never made it to the public). The vulnerability
history of php5 indicates that the component is relatively hardened at the time the next version is released.
The vulnerability
discovery rate is relatively low and stable for the last months before the launch of php7.
\begin{figure}[h]
\centering
\includegraphics{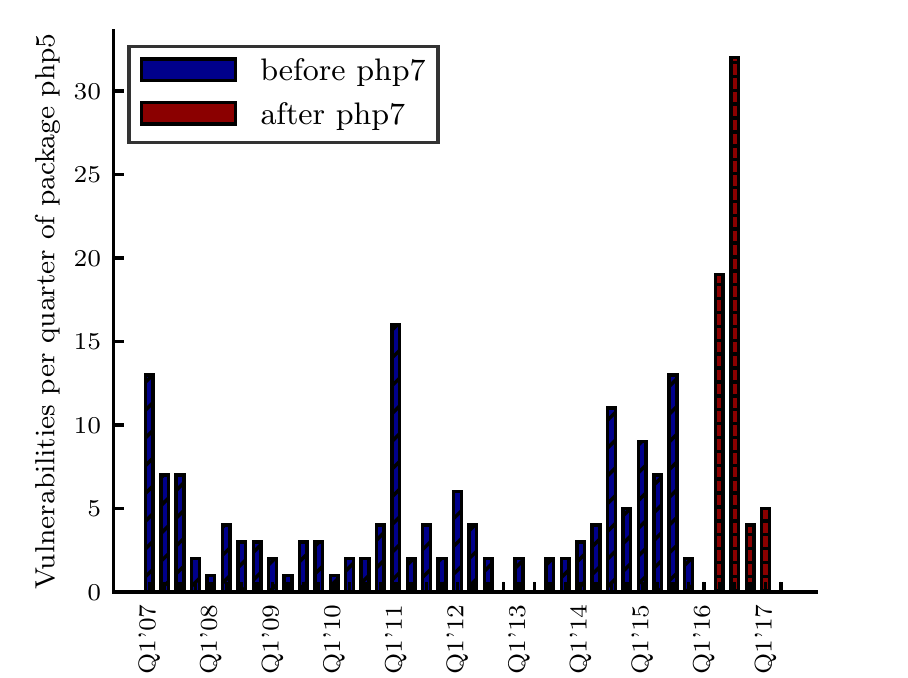}
\caption{Vulnerabilities of package php5, during its presence in the stable release, before and after the introduction of the next version (php7) in testing.
	Vulnerability rate: (a) before the launch of the new version: $\approx 4$ vulnerabilities/quarter;
(b) after the launch of the new version: $\approx 10$ vulnerabilities/quarter. }
\label{fig:php}
\end{figure}
According to the hypothesis, we expect
a good amount of vulnerabilities after this point to affect the new version (php7) of the software and not older versions (php5.x).
On the contrary, we observe that indeed we have a substantial hike in vulnerabilities reported due to the launch of the new
version, however, most of those vulnerabilities were there from the previous version. The launch of the new version
may have instigated researchers and bug hunters to look for vulnerabilities induced by the new code, but instead what they found
were already existing vulnerabilities from previous versions. After manual inspection of all security incidents tracked by the Debian
Security Bug Tracker\footnote{https://security-tracker.debian.org/tracker/}, which also tracks vulnerabilities of packages in the testing
distribution, we found that in the time window of January~2016 - November~2017, out of the 93 vulnerabilities that affected php7, 78 of them
(84\%) also affected version 5 of the software (4 of the 78 did not affect version 5.4 and only affected version 5.6).
\medskip
\par\noindent\textbf{OpenJDK:}\smallskip\\
OpenJDK is an open source implementation of the Java Platform (Standard Edition), and
since version 7, the official reference implementation of Java.
We repeat the experiment performed with PHP, with OpenJDK versions 7 and 8.
Version 7 was introduced into the testing distribution of Debian in September 2011 and became part of the stable
distribution in May 2013 (Debian Wheezy). It remained part of the stable until the release of stretch (June 2017).
The next version, OpenJDK-8, became part of the testing distribution in May 2015 and became part of stable with
Debian stretch (June 2017), replacing version 7.
\begin{figure}[h]
\centering
\includegraphics{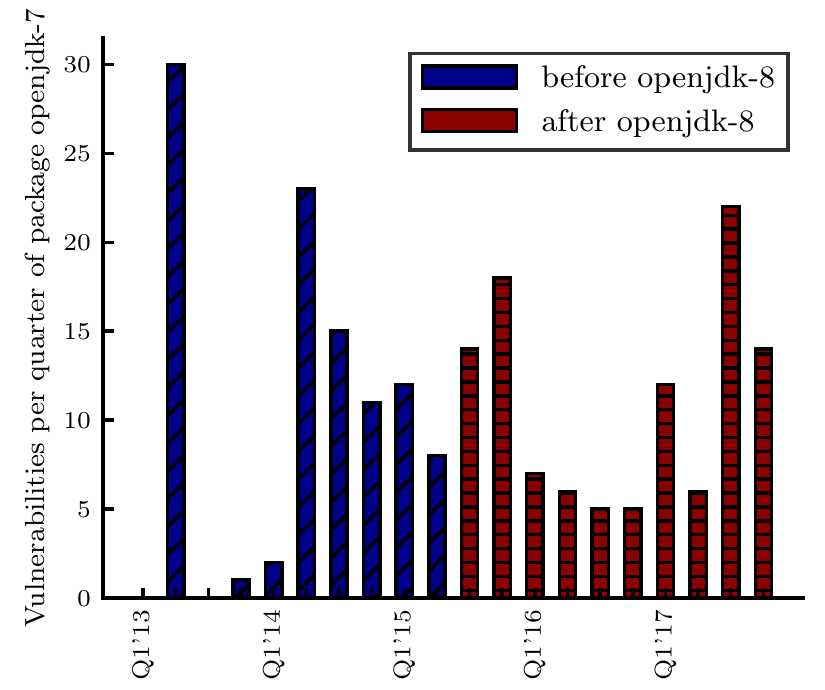}
\caption{Vulnerabilities of package openjdk-7, during its presence in the stable release, before and after the introduction of the next version (openjdk-8) in testing.
Vulnerability rate: (a) before the launch of the new version: $\approx 11.3$ vulnerabilities/quarter;
(b) after the launch of the new version: $\approx 10.6$ vulnerabilities/quarter. }
\label{fig:openjdk}
\end{figure}
In Figure~\ref{fig:openjdk}, we see the vulnerabilities of version 7 before and after the introduction of the next version. Again,
we see no significant decline in the rate of vulnerability reports, and the introduction of the
next version seems to contribute to the discovery of vulnerabilities of the previous version. To put things into perspective,
out of a total of 38 vulnerabilities that were reported for openjdk-8 in the time span of June-November 2017, only 2 did not affect
version 7 (although these are not depicted in Figure \ref{fig:openjdk} because version 7 was removed from stable in the meantime),
and most of them (31/38) also affected version 6, which was introduced in the testing distribution in 2008.
\medskip
\par\noindent\textbf{Debian Wheezy:}\smallskip\\
Although, the detailed investigation of vulnerabilities for PHP and OpenJDK gave us
some useful insights about the current state of software quality, these results
cannot be generalized to other packages. In order to get a more complete view of
the effect of new vulnerabilities on older versions, we study the security history
of a single stable release of Debian, including its Long-Term Support (LTS) phase\footnote{
	Starting from 2014, Debian includes an LTS program, in order to extend support
for any release to at least 5 years in total.}.  For this study, we chose Debian 7
(Wheezy) that was release in May 2013 and at the moment of writing is still supported
from the LTS team (planned support until May 2018). In Figure~\ref{fig:wheezy}, we see the
distribution of vulnerabilities per quarter, starting from the release of Wheezy, until
the time of writing (November 2017).
\begin{figure}[h]
\centering
\includegraphics{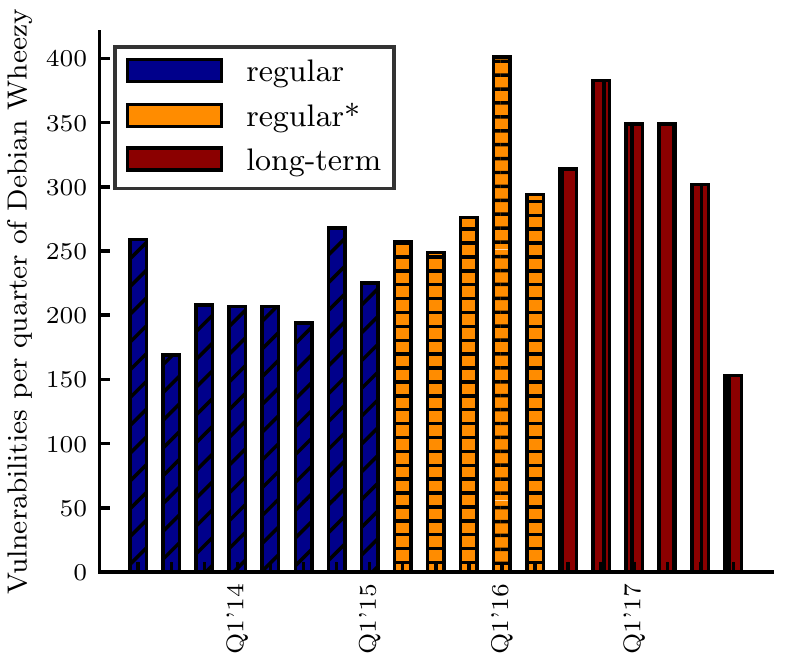}
\caption{Vulnerabilities that affected packages of the Wheezy Debian release. *During the
	time window from Q2/2015 to Q2/2016, both Debian 7 and 8 were supported by the
regular security team, and therefore, the amounts presented in the figure are a higher
bound, as some vulnerabilities may have affected only the newer release.}
\label{fig:wheezy}
\end{figure}
Even for a specific stable release of Debian, we can observe a clear upward trend that
continues in the LTS phase of the software. These results support our findings for
individual packages and show that the rate of vulnerabilities is not decreasing, and to
the contrary is slightly increasing over time, even for a specific stable release.
\medskip
\par\noindent\textbf{Conclusions:}\smallskip\\
In this section, we investigated the distribution of vulnerabilities in the Debian
ecosystem. After detailed examination of the vulnerabilities reported for both individual
widely used packages (PHP and OpenJDK), and a specific stable release of Debian, we
conclude that the number of vulnerabilities does not decrease over time, even for
software that has been stable for many years. To the contrary, we discerned a relatively
stable rate of vulnerabilities, that shows signs of increase over time.
In other words, we are still in the phase of \emph{the more we look - the more we find}.
Maybe our vulnerability-finding efforts are like cutting off the tip of an iceberg;
it rises up a little, but we are not making any visible progress.
Although automated security tools and manual security inspection are becoming more
widespread and effective, we have not reached the point of curbing the vulnerability
rate. Our results draw interesting comparisons to
studies performed over a decade ago. Rescorla claimed in \cite{rescorla2005finding}
that there was no clear evidence that finding vulnerabilities made software more
secure, and that even the opposite may be true, i.e. that finding vulnerabilities,
given that their rate is not decreasing, leads to more risk than good, by
allowing hackers to attack unpatched systems. Another study from 2006
by Ozment and Schechter~\cite{ozment2006milk} found evidence of a decrease in the
vulnerability rate of the base OpenBSD system in a 7.5 year interval. Our results
show that, more than a decade later, this is not generally true for Debian as a whole
and for PHP and OpenJDK individually. After the impressive growth of the security community
since 2006, we still do not have strong evidence that the security quality of software
is increasing.
\begin{remark}
	The Debian Security Team publishes DSAs for important vulnerabilities that command immediate patches of the packages. These are a subset
	of the vulnerabilities that have an associated CVE. Therefore, the numbers presented in the above section, can generally be seen as a
	lower bound. Furthermore, the security team does not differentiate between the
	vulnerabilities that were deemed important enough to command a DSA, e.g. by
	using a CVSS severity score~\cite{mell2006common}.
	In our analysis, we also did not take into account the severity score of
	vulnerabilities as reported in CVEs, preferring to follow the judgement of
	the Debian Security Team on which vulnerabilities need an urgent fix.
	However, using severity data for vulnerabilities may allow us to draw
	other interesting conclusions in the future.
\end{remark}

%
\section{Predicting future vulnerabilities}\label{sec:pred}
Our aim is to predict future zero-day vulnerabilities of software components (in our instantiation Debian packages),
based on their vulnerability history, mined from public and open
vulnerability databases. In this section, we present a formalization of the
problem via an abstract functionality, along with the realization of this functionality by
three different prediction techniques. We then proceed to compare those techniques on
real-world data and discuss their performance.
\begin{boxfig}{The Prediction functionality}{The abstract prediction functionality.}{fig:abstractPred}
	\par \textbf{-- Input: }
	\begin{enumerate}[i.]
		\item A list of time-series $TS = \{ts_{1},\cdots,ts_{n}\}$,
			where $n$ is the total number of software components. For each
			time-series: $ts_{i, \: i \in \{1,\cdots,n\}}  =  \langle v_{(i,1)},\cdots,v_{(i,m)} \rangle$, 
			where $v_{(i,j)}$ is the number of vulnerabilities of component $i$ for the time
			window $t_j$.
		\item A list of validation samples $V = \{val_{1},\cdots,val_n\}$, where
			$val_{i, \: i \in \{1,\cdots,n\}}  =  \langle v_{(i,m+1)},\cdots,v_{(i,m+k)} \rangle$.
		\item An expected prediction window $l$.
	\end{enumerate}
			\medskip
	\par \textbf{-- Output: }
	\begin{enumerate}[i.]
		\item A list of predictions $Pred = \{pred_1,\cdots,pred_n\}$,
			\medskip \\
			where $pred_{i, \: i \in \{1,\cdots,n\}} = \widehat{\sum\limits_{j=1}^{l}v_{m+k+j}}$.
			\smallskip \\
		\item A list of error estimates $Errors = \{e_1,\cdots,e_n\}$, where
			each $e_{i, \: i \in \{1,\cdots,n\}} = \widehat{error(pred_i)}$, where $error()$ is an error metric, e.g. absolute distance, normalized root mean squared error (nrmse) etc.

	\end{enumerate}
\end{boxfig} 
\subsection{Problem statement and experimental setting}
The abstract formulation of our problem can be seen in Fig. \ref{fig:abstractPred}.
On input, a list of time series $TS$, containing the vulnerability histories of
the software components under study for $m$ time windows, the validation size (in time steps)
$k$, and the expected prediction window into the future $l$, the functionality outputs a
list of predictions $Pred$ and a list of corresponding \emph{future} error estimates $Errors$. Our goal is
to predict the number of vulnerabilities in well-defined future intervals, e.g. the
expected number of vulnerabilities of a software component in the next 9 months.
Our dataset consists of vulnerability data of Debian packages, as already mentioned.
Vulnerabilities are grouped in a per-month basis, with all data of the years 2001-2016
as input, and data for the year 2017 as output. In this section, we study and compare
3 prediction techniques, namely (i) a simple average function on the history of the package,
(ii) a more advanced exponentially weighted average giving higher weights to the most recent
history of the packages, and finally (iii) an \ac{LSTM} neural network, which represents the
state of the art in machine learning techniques for on time-related data.

\smallskip
\par\noindent\textbf{Average: }The simplest prediction technique we implemented is the average
function on the vulnerability history of the packages. The intuition behind this approach is
that software components have some characteristics that define their security behaviour and
these characteristics are generally stable. Additionally, our observations in the detailed
analysis performed in the previous section, reinforce this intuition.

\smallskip
\par\noindent\textbf{Weighted average: }The second method we implemented is an exponentially
weighted average function. The weighting allows us to take into account changes to the behaviour
of the package that happen over an extended period of time.

\smallskip
\par\noindent\textbf{LSTM: }By using \acp{LSTM} we aspire to capture changes to the behaviour of
software components with fine granularity. We train a separate \ac{LSTM} model for each package
by using the packages vulnerability history from Q1/2001-Q1/2016 as the training set, the data  from
Q2-Q4/2016 as the validation set which also produces the future error estimate, while we make predictions
for time windows in 2017 (Q1-Q3), as seen in Figure~\ref{fig:timeline}. We implemented our prediction module in around 400 lines of python code,
by using the Keras\footnote{https://keras.io/} open source neural network library.
The LSTM that provided the predictions that follow is a stateful one,
consisting of 10 neurons. We divided the vulnerabilities in monthly intervals and applied a rolling average
function, as a pre-processing step of the input time series. We also configured
the network to look 3 steps (months) in the past, in order to generate a prediction. We trained 5 models
for each package and selected the one that yielded the minimum validation interval error, for our
final predictions.
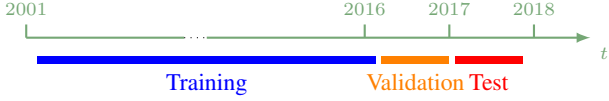
\begin{figure}
\begin{tikzpicture}[x=1.5cm]
\draw[mygreen,thick]
  (0,0) -- (1.40,0);
\draw[dotted]
  (1.40,0) -- (1.60,0);
\draw[mygreen,->,thick,>=latex]
  (1.60,0) -- (5,0) node[below right] {$\scriptstyle t$};
  \draw[mygreen,thick]
  (0,0) -- ++ (0,5pt) node[above] {$\scriptstyle 2001$};
  \draw[mygreen,thick]
  (3,0) -- ++ (0,5pt) node[above] {$\scriptstyle 2016$};
  \draw[mygreen,thick]
  (3.75,0) -- ++ (0,5pt) node[above] {$\scriptstyle 2017$};
  \draw[mygreen,thick]
  (4.50,0) -- ++ (0,5pt) node[above] {$\scriptstyle 2018$};

  \fill[blue] (0.1,-0.25) rectangle node[below] {\strut\small Training} (3.1,-0.35);
  \fill[orange] (3.15,-0.25) rectangle node[below] {\strut\small Validation} (3.75,-0.35);
  \fill[myred] (3.80,-0.25) rectangle node[below] {\strut\small Test} (4.40,-0.35);

\end{tikzpicture}
	\centering
	\caption{Experiment timeline partition}
	\label{fig:timeline}
\end{figure}

\subsection{Prediction Results}
In Table~\ref{tab:lstm}, we see the prediction results
of the LSTM method for the 10 most vulnerable packages of all time that are part of the Debian \emph{stretch} (stable) distribution.
To judge the prediction accuracy of the neural network, we compare the root mean squared error (rmse) of the method over the top
136 packages (packages that have more than 10 vulnerabilities reported), with the other proposed methods, namely the average function and the weighted average function.
As we can see in Table~\ref{tab:comparison},
the LSTM achieves an rmse of $14.66$, whereas the next most accurate method, the weighted average one, achieves an error of $18.07$. This translates to $19\%$ better accuracy
for the neural network implementation, or inversely $23\%$ larger error for the weighted average technique, which we consider substantial.
Additionally, as the optimization of the neural network was considered out of the scope of this paper, we believe that
there is still potential for significantly improved accuracy. However, we also believe that there is a natural bound to the
accuracy of the predictions we can generate, as the vulnerability discovery process has some inherent unpredictability.
The required error estimate for a component (package) $i$ (see Figure~\ref{fig:abstractPred}) is calculated as the normalized error of the prediction for the
\emph{validation} period, i.e. the last 9 months of 2016 in this case, as follows:
\begin{equation}
	error(pred_i) = \left\{
\begin{array}{ll}
      \frac{|\widehat{v_{i,m+k}} - v_{i,m+k}|}{v_{i,max}-v_{i,min}} & ,v_{i,max}-v_{i,min} > 1 \\
      \\
      |\widehat{v_{i,m+k}} - v_{i,m+k}| & ,else \\
\end{array} 
\right. 
\label{eq:cert}
\end{equation}
%
%
\begin{table}[htbp]
	\centering
\caption{Prediction results of LSTM model - prediction window is 9 months (January-September 2017)}
\begin{tabular}{|l|c|c|c|c|}
\hline
Package & Rank & Predicted & Actual & Future Error exp.\\
\hline \hline
linux & 1 & 35 & 70 & 0.034\\ \hline
firefox-esr & 2 & 54 & 84 & 0.08 \\ \hline
chromium-browser & 3 & 80 & 74 & 0.293 \\ \hline
openjdk-8 & 4 & 14 & 42 & 0.196 \\ \hline
icedove & 5 & 24 & 58 & 0.408 \\ \hline
wireshark & 6 & 12 & 11 & 0.136 \\ \hline
php7.0 & 7 & 14 & 5 & 0.821 \\ \hline
mysql-transitional & 8 & 11 & 32 & 0.136 \\ \hline 
openssl & 9 & 3 & 3 & 0.491 \\ \hline
qemu & 10 & 6 & 11 & 0.674 \\ \hline
\end{tabular}
\label{tab:lstm}
\end{table}
\begin{table}[htbp]
	\centering
\caption{Root mean squared error of different prediction techniques on the top 136 vulnerable packages of Debian}
\begin{tabular}{|l|c|c|}
\hline
Technique & rmse & error w.r.t. best \%\\
\hline \hline
LSTM & 14.66 & -- \\ \hline
Average & 18.65 & +27\% \\ \hline
Weighted average & 18.07 & +23\% \\ \hline
\end{tabular}
\label{tab:comparison}
\end{table}
\subsection{Remarks}
We conclude this section with several remarks regarding the prediction methodology and results.
\smallskip
\par\noindent\textbf{Threats to validity:}
We decided to take into account data only regarding the past security incidents of a software component.
Therefore, we overcame threats like models generating good predictions only for certain programming languages
or types of components. However, the ground truth at our disposal consists only of the vulnerabilities
reported and in some cases this set of vulnerabilities might be different with the set of vulnerabilities
discovered. Still, we believe that a transparent and open software management environment is a good direction
towards making the two aforementioned sets converge.
\smallskip
\par\noindent\textbf{Computational performance:}
Regarding the time and computational resources required to generate our predictions, as expected the average
and weighted average approaches required negligible resources, while the deep neural network deployed
required on average 5.3 minutes to generate and apply the models for each package running on a commodity laptop (Intel i5 cpu, 8\,GB RAM).
Note that re-computation of the assessments need only be performed when significant new data becomes available
and can be scheduled, e.g. in a monthly basis. Therefore, even the overhead incurred by the neural network can
be considered very small.

%
\section{Software Trustworthiness Model}\label{sec:model}
In order to assess the security quality of a software system, we propose a trust model that considers past
behaviour of underlying components as well as their inter-dependencies within that system. The model is
rooted from an extended version of Bayesian statistics, namely CertainTrust and CertainLogic
(see Section~\ref{sec:comptrust}). Our proposed model considers the probability estimate of software package
vulnerabilities and the inherent certainty of the estimated probability as inputs to the CertainTrust
representation (see Section~\ref{sec:comptrust}). This means that the probability estimate (regarding vulnerabilities)
of a software component can be mapped to the parameter, $t$, the certainty of the aforementioned probability estimate can be mapped to
the parameter, $c$, and prior knowledge about the software component under assessment can be mapped to the parameter,
$f$. In the upcoming sections, we formally devise our probabilistic model for software component trustworthiness, as
well as for the assessment of complex software systems with inter-dependent components.
\subsection{Single component trust model}
In this section, we model trust for individual software components. First, we define the \emph{quality} of a software component.
Second, we define the expected \emph{trustworthiness} of a component, consistent with Bayesian statistical inference, and compatible with
the parameters of the CertainTrust representation.
Last, we prove that the expected trustwothiness assigned to a software component is an optimal estimator of the quality of that component.
%
\begin{definition}[Software component quality]
	We define the \textit{quality} $Q_{i,t_p}$ of a software component $i$ as the probability
	that this component will not be found vulnerable in the next well-defined time period of $t_p$ time steps.
	The complementary probability $1-Q_{i,t_p}$ is the security \textit{risk}
	$R_{i,t_p}$ associated with the component.
	\label{def:trustworth}
	\begin{equation}
		\label{eq:trust_def}
		\begin{split}
		Q_{i,t_p} = Pr[i\text{ not vulnerable}]\\ \text{ for } t\in(t_{now},t_{now}+t_{p})
		\end{split}
	\end{equation}
	\begin{equation}
		R_{i,t_{p}} = 1 - Q_{i,t_{p}}\\ \text{ for } t\in(t_{now},t_{now}+t_{p})
	\end{equation}
\end{definition}
\smallskip
Our goal is to assess the quality of a component in future time intervals.
We use the predictions generated by the prediction module of our tool as a point estimate
for the number of vulnerabilities. Our prediction module predicts the number of
vulnerabilities a software component $i$ will have in the next $l$ time steps, with $l= \lambda t_p,\, \lambda \in \mathbb{N} \text{ and } \lambda \gg1$.
The time period of $l$ (e.g. $l = 9$ months in our use-case) is relatively small in comparison to the total
history of a software component and therefore we assume that vulnerabilities follow a
Binomial distribution (Bernoulli process) with $t_p$ time steps between each trial. Note that this assumption
is used solely for modeling purposes in order to construct a measure of trustworthiness.
Assuming a Bernoulli trial each $t_p$, we can estimate the quality of
a component by the following:
\smallskip
\begin{equation}
	\label{eq:comp_trust}
	T_i = 
	\left\{
	\begin{array}{ll}
		1 - \frac{pred_i}{\lambda}  & \mbox{if } pred_i \leq \lambda \\
		0 & \mbox{if } pred_i > \lambda
	\end{array}
\right.
\end{equation}
where $pred_i$ is the total predicted number of vulnerabilities for the next $l$ time steps,
and the $t_{p}$ argument is suppressed in the notation.
The second part of equation \ref{eq:comp_trust} is included for
completeness reasons. In a system deployment scenario (e.g. see section~\ref{sec:practice}),
the parameter $\lambda$ (or equivalently the parameter $t_p$) is set to a value
that practically makes it impossible to have $pred_i \leq \lambda$.

We now proceed to define the \emph{trustworthiness} of a component, which is an
expectation value about the its quality.
\begin{definition}[Software component trustworthiness]
	We define the \textit{trustworthiness} of a software component $i$ as the
	expectation $E(t,c,f)=t \cdot c + (1-c) \cdot f$ associated with a
	CertainTrust tuple $(t, c, f)$, where $t\in[0,1]$ is an optimal point estimation
	of the component's quality, $c \in [0,1]$ is
	a certainty value for this estimation, and $f$ is a calibrating
	factor stemming from a priori knowledge about the component.
	\label{Defn:softTrust}
\end{definition}
\smallskip
We take $t$ to be equal to $T_i$ from equation \ref{eq:comp_trust}.
The certainty value $c$ is derived from the prediction mechanism
used in the scheme. Inspired by \cite{hauke2013application}, in our implementation we
use the normalized error of the prediction model for the validation phase, as a
conservative goodness-of-fit measure (see Figure \ref{fig:abstractPred}).
Thus, the certainty estimate for a component $i$ is calculated as:
\begin{align}
	&c = 1 - min(error(pred_i, 1))
	\label{eq:certainty}
\end{align}
where $error(pred_i)$ is a normalized error estimate. In our setting, this value is
calculated as per Equation~\ref{eq:cert}.
\begin{theorem}
	Assuming vulnerabilities generate according to a Bernoulli process during a time period
	of $l$ time steps, then equation \ref{eq:comp_trust}, expressing the trustworthiness of a
	software component $i$, is a \ac{MLE} of the quality of the component, as
	defined in equation \ref{eq:trust_def}.
	\label{theorem:comptrust}
\end{theorem}
\begin{proof}
	The average number of vulnerabilities per $t_p$ time steps, during a time period of $l$, $l \gg t_p$ time steps,
	$\frac{pred_i}{\lambda}$, is the Maximum Likelihood Estimator (MLE) for the probability of success of each Bernoulli trial.
	We have:
	\begin{displaymath}
		\frac{pred_i}{\lambda} \estimatesB Pr[\text{trial succes}] = Pr[\text{i vulnerable}]
	\end{displaymath}
	\smallskip
	\\Thus, the complementary probability, $T_i$ is an MLE of the component's quality
	$Q_{i,t_{p}}$ as expressed in equation
	\ref{eq:trust_def}:
	\begin{displaymath}
		\begin{split}
		T_i = 1 - \frac{pred_i}{\lambda} \estimatesB  Pr[\text{i not vulnerable}] \\\text{for } t\in(t_{now},t_{now}+t_{p})
		\end{split}
	\end{displaymath}
\end{proof}
\subsection{Software system trust model}\label{sec:combining}
Most modern computing systems consist of multiple software components. These components can depend on each other,
or they can be configured in a way to give the system redundancy, i.e. to allow the system to uphold its security
guarantees even if one of the components is compromised. An example of the latter is a private database where entries
are secret-shared~\cite{shamir1979share} between two or more machines. The security dependencies found in a
software system can be depicted in a graph, similar to that of Figure~\ref{fig:graph}. The graph can be straightforwardly
expressed via a propositional logic formula. Setting the atomic formula $i_{i\in\{A,B,C,D,X,Y\}}$ to model that the software component $i$ is safe, i.e.
it is not vulnerable, the resulting propositional formula is:
\begin{equation}
	SYSTEM \logeq B \land D \land [(A \land C) \lor ( X \land Y)]
	\label{eq:proplog}
\end{equation}
Note that the software components are assumed independent at this stage, i.e. they do not share code. A well-implemented
``divide and conquer`` strategy, like the one enforced by Debian packaging, would satisfy this assumption.
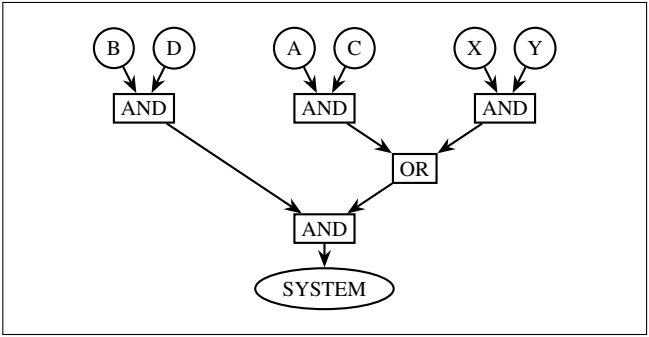
\begin{figure}
\caption{Security dependencies graph of a complex system.}
\label{fig:graph}
\smallskip
\begin{framed}
\centering
\begin{tikzpicture}[scale=0.8]
\begin{scope}[every node/.style={circle,thick,draw,scale=0.8}]
	\node (A) at (0,3) {A};
    	\node (B) at (-3,3) {B};
    	\node (C) at (1,3) {C};
    	\node (D) at (-2,3) {D};
    	\node (X) at (3,3) {X};
    	\node (Y) at (4,3) {Y} ;
\end{scope}
\begin{scope}[every node/.style={rectangle,thick,draw,scale=0.8}]
	\node (AND1) at (-2.5,2) {AND};
	\node (AND2) at (0.5,2) {AND};
	\node (AND3) at (3.5,2) {AND};
	\node (OR1) at (2,1) {OR};
	\node (AND4) at (0.5,0) {AND};
\end{scope}
\begin{scope}[every node/.style={ellipse,thick,draw,scale=0.8}]
	\node (S) at (0.5,-1) {SYSTEM};
\end{scope}
\begin{scope}[>={Stealth[black]},
              every node/.style={fill=white,rectangle,scale=0.6},
              every edge/.style={draw=black,thick}]
	      \path [->] (B) edge (AND1);
	      \path [->] (D) edge (AND1);
	      \path [->] (A) edge (AND2);
	      \path [->] (C) edge (AND2);
	      \path [->] (X) edge (AND3);
	      \path [->] (Y) edge (AND3);
	      \path [->] (AND2) edge (OR1);
	      \path [->] (AND3) edge (OR1);
	      \path [->] (AND1) edge (AND4);
	      \path [->] (OR1) edge (AND4);
	      \path [->] (AND4) edge (S);
\end{scope}
\end{tikzpicture}
\end{framed}
\end{figure}
CertainLogic's $AND^{ct}$ and $OR^{ct}$ operators (see Definitions \ref{Defn:CT_AND} and \ref{Defn:CT_OR}) can be used
to generalize the propositional logic operators of Equation~\ref{eq:proplog}, leading to our system trustworthiness
definition.
\begin{definition}[Software system trustworthiness]
	The trustworthiness of a system $S$, whose security dependencies can be expressed by
	a propositional logic formula with no variable repetition (like the one in \ref{eq:proplog}) is defined as the evaluation
	of the formula with the propositional logic terms substituted by CertainTrust terms and the propositional logic
	operators substituted by CertainLogic operators. In relaxed mathematical notation:
	\begin{equation}
		\label{eq:systrust}
	\begin{aligned}
		\text{if } S \logeq F(a_1,\cdots,a_n) \text{ , then }\\ T_S = F[a_i^{ct}/a_i,\land^{ct}/\land,\lor^{ct}/\lor]
	\end{aligned}
	\end{equation}
	where $F$ is a propositional logic formula with variables $a_1,\cdots,a_n$, which can be brought to a form, so that
	each variable appears once. This constraint stems from the fact that CertainLogic operators, similarly to Subjective
	Logic ones, are designed to operate on independent propositions.
	\label{def:systrust}
\end{definition}
\begin{theorem}
	The value assigned to a software system by the evaluation of equation~\ref{eq:systrust} is
	a valid expectation of the quality of the system (seen as a component),
	as defined in \ref{eq:trust_def}. Specifically, the trust value $t$ is a maximum a posteriori probability (MAP) estimate of the
	quality of the system.
\end{theorem}
\begin{proof}
	From Theorem~\ref{theorem:comptrust} we have that for each component $a_1,\cdots,a_n$,
	its trustworthiness is a valid expectation of its quality, i.e. the trust value
	$t$ is an MLE of the component's quality.
	CertainLogic's $AND^{ct}$ and $OR^{ct}$ operators are consisted with Bayesian statistics, and therefore provide
	MAP estimates for the degree of truth of (a) both proposition simultaneously and (b) of at least one of the two propositions, respectively.
	Therefore, the evaluation of a CertainTrust logical formula provides an MLE for the truth value of the underlying
	reasoning that is in accordance with propositional logic rules.
\end{proof}
Having established a systematic method for calculating the trustworthiness of complex software systems, starting from a
graph representation (or the equivalent propositional logic formula), we proceed to examine the issue of
extracting the aforementioned graph representation from real-world software systems. For example take
the case where the overall system is a database with data secret-shared among two sub-systems. Due to the
secret sharing technique, both sub-systems would need to be compromised in order for the overall system to
be compromised, i.e. for the data to be leaked. The naive solution of calculating a trust score for each
subsystem and then combining the two scores using CertainLogic's $OR^{ct}$ operator would only produce
a meaningful result if the two sub-systems did not share any software components. This, in general is not
the case. Consequently, it is important to carefully construct the system dependency graph before proceeding
with the trustworthiness calculation. To follow on our established example of data entries secret-shared between
two sub-systems, let these subsystems consist of the following components: $S:\{S_1, S_2\}$, with $S_1:\{A, B, C, D\}$ and $S_2:\{B, D, X, Y\}$.
Following the notation of this section, the resulting propositional logic formula for the security of the
system as a whole would be: $S = (A\land B\land C\land D)\lor(B \land D \land X \land Y)$, leading to the
simplified formula already presented in Equation~\ref{eq:proplog} and Figure~\ref{fig:graph}. Take note
that substituting the propositional logic operators with CertainLogic's counterparts in the initial formula
would not be possible, due to the appearance of variables more than once. Although the assumption of being able
to express the propositional logic formula in a form without variable repetition was satisfied in this
example, it is not always the case. In the case where such a simplification is not possible, we follow
an approach, a variation of which was shown to be optimal in~\cite{josang2008optimal}. Specifically,
we express the formula in its disjunctive normal form and proceed to delete terms until the formula can
be expressed with no variable repetition. The criteria with which terms are deleted are (a) the resulting
formula has the least number of variable repetitions, and (b) the term deleted would be the one with the
lowest certainty value if its CertainTrust representation was calculated. The technique described above
is a conservative approach erring on the side of caution, meaning that the resulting formula will be harder to satisfy, and thus the
resulting score should be considered a lower bound on the quality of the system.
%
\medskip\par\noindent\textbf{Fusion of assessments from different sources: }
It is generally the case that different combinations of Data and Prediction modules (see Figure~\ref{fig:arch}) yield
different results for the trustworthiness of the same software component. For example, an anti-virus company could
have its own database of security incidents, in addition to the publicly available ones. In addition, it could use a different
prediction technique, e.g. one that includes static analysis of the software components. A system administrator should be able
to incorporate the knowledge provided by this source into the trustworthiness opinion they have already computed using the
means available to them. To this end, CertainLogic's fusion operators (see Appendix~\ref{ap:fusion}) can be used to combine
opinions about the same software components. These operators have been designed to model e.g., the scenario where a person
gets conflicting recommendations about a product in an online marketplace. The parallelism to our scenario is straightforward,
and thus these operators naturally fit our use-case, both from a mathematical and a sociological point of view.
\subsection{Tuning the model}\label{sec:practice}
There are some decisions to be made, concerning the way our model is going to be applied in
a real-world setting.  The model parameters need to be set or bounded by
empirical values, so as to have results that closely depict reality.
\smallskip\par\noindent\textbf{Setting the parameter $\boldsymbol\lambda$: }
The parameter $\lambda$ is recommended to be set, so that: $$\lambda > \sum\limits_{i=1}^{n} pred_i$$
This way, the second part of Equation~\ref{eq:comp_trust} will not be activated, even when considering
the worst case scenario of a single system that depends on all the components that are predicted
to be vulnerable. In our scenario, where $l$ is the equivalent of 9 months, we set $\lambda=4\cdot30\cdot 9=1080$, which effectively means that the expectation
value produced by M-STAR is the probability that a system will be found valuable, sampled at intervals of six hours.

\smallskip\par\noindent\textbf{Limiting value ranges for estimates: }
All three model parameters ($t$,$c$,$f$) are probabilities, therefore they live in
the real interval $[0,1]$. However, in a real-world deployment of M-STAR we may want
to limit the range of the values to a subinterval of $[0,1]$, in order to avoid corner cases
and produce better results. Specifically, we limit the certainty estimate range to $[0.100,0.990]$,
with the following reasoning. First, due to the normalized error calculation formula,
for packages that have very few vulnerabilities, a reasonably good prediction can lead to a certainty of $0$.
Second, even if a model fits reality perfectly (the error is very close to $0$) in the validation
interval, it is likely that this will not be the case for the future prediction, and thus we assume
a possible error margin of $0.01$.
\smallskip\par\noindent\textbf{Setting prior knowledge value (f): }
The \emph{a priori} expectation of the quality of a component is set empirically, based on observations
we have made on the Debian ecosystem as a whole. Due to the power-law like distribution of vulnerabilities
among Debian's components (see Figure~\ref{fig:distr}), we decide to partition the dataset into two, when setting the prior knowledge value.
First, for the top 20 vulnerable packages, which represent the dominating subset, we set their initial expectation,
as the average number of vulnerabilities of this group during a two-year interval (2015-2016). For the packages
that had at least one vulnerability in their history, we set their initial expectation to the normalized average
number of vulnerabilities of those packages in the same two-year interval (2015-2016). Regarding the two cases mentioned above,
we additionally apply a scaling factor, accounting for the global observations we have made showcasing that the average
number of vulnerabilities per package is increasing, and computed by fitting a first order polynomial on the data of
Figure~\ref{fig:average_per_year}. The final value is therefore $f' = 1.05\cdot f-0.05$. For the packages that have
no reported vulnerabilities, we set the initial expectation to $1$, i.e. we consider them fully trustworthy.
Apart from the empirical solution provided here, there could be other options, e.g. set the initial value to
$0.5$ globally (non-informative prior), or set the initial value by performing static analysis on the code of the components and pinpointing high-risk
ones.
%
%

%
\section{Visualization and deployment}\label{sec:apl}
\par\noindent\textbf{Visualization: }
The visual representation of calculated trust and certainty values is essential to actually aid users in decision-making processes.
Consequently, we based the design of our visualization module on T-Viz~\cite{volk2014communicating}, a tested foundation evaluated in various user studies.
We migrated T-Viz to the field of software security and added a representation of the security history of
the component as help to the system administrator. In Figure~\ref{fig:example}, we can see the trustworthiness assessment of our full-fledged
installation, focusing on the components that have the lowest trustworthiness score (expectation). The lengths of the slices
represent the trust values calculated for each component, their width show the associated certainty, their color
characterize the expectation, and the value in the middle the trustworthiness (expected quality) of the system as a whole.
The slices are clickable and produce a detailed report on the component, including the smoothed time-series
of the actual vulnerabilities (blue), as well as the time-series produced by the prediction mechanism (red).
Although we conducted informal interviews that generated positive comments about the visualization of Figure~\ref{fig:example},
an extensive user study would be beneficial. However, this is out of scope of this paper.
\\\par\noindent\textbf{Deployment: }
After going through the empirical and theoretical foundations of M-STAR,
we proceed to showcase the utility of our system by providing sample use-cases.
First, we compare the trustworthiness of two Debian packages, namely \emph{firefox-esr}
and \emph{linux} (kernel package).
Then, we compare the trustworthiness of two systems,
where one is a full-fledged Debian installation for general use and the other is
a web server. Finally, we assess the security of a fictional
system comprised of the two aforementioned systems, in a configuration, where
there is a 1+1 redundancy. The results are summarized in Table~\ref{tab:results}.
The first three columns show the computed parameters of the CertainTrust model,
the fourth the resulting expectation score, followed by an ``equivalent'' expected
number of vulnerabilities, the actual number, and the ratios of the equivalent and
actual vulnerabilities along with their normalized error. The latter expresses the error
in the expectation about the relative quality of two configurations.
We observe that M-STAR assesses the relative quality of the two systems with
an error of less than 10\%. The calibrating parameters of the system can be modified
over time by the user\,/\,administrator, so the equivalent number of vulnerabilities estimates the actual
number more closely.
However, the issue of accurately predicting the absolute number of vulnerabilities requires further
investigation.
For the 1-out-of-2 scenario, we observe that there is a rather small 7\% decrease in the equivalent
vulnerability number w.r.t. the single web server, which is expected as the two systems have a lot of components in common.
In this case, more software diversity would be required to achieve a better security level.
\begin{table*}[htbp]
	\centering
\caption{Trust assessment of different configurations}
\begin{tabular}{|l|c|c|c|c|c|c|c|c|c|}
\hline
& t & c & f & expectation & equivalent number & actual number & \textbf{ratio equivalent} & \textbf{ratio actual}& \textbf{ratio norm. error}\\
\hline \hline
linux & 0.968 & 0.966 & 0.974 & 0.968 & 35 & 70 & \multirow{2}{*}{\textbf{0.673}} & \multirow{2}{*}{\textbf{0.833}}& \multirow{2}{*}{\textbf{0.192}}  \\ \cline{1-7}
firefox & 0.950 & 0.920 & 0.974 & 0.952 & 52 & 84 & & &\\ \hline
Full-fledged & 0.687 & 0.662 & 0.502 & 0.625 & 405 & 809 &\multirow{2}{*}{\textbf{1.770}} & \multirow{2}{*}{\textbf{1.954}}& \multirow{2}{*}{\textbf{0.094}} \\ \cline{1-7}
Web server & 0.840 & 0.690 & 0.673 & 0.788 & 229 & 414 & & &\\ \hline
1-out-of-2 & 0.842 & 0.693 & 0.710 & 0.802 & 214 & -- & -- &-- & --\\ \hline
\end{tabular}
\label{tab:results}
\end{table*}
\begin{figure}
\centering
\includegraphics[width=\columnwidth]{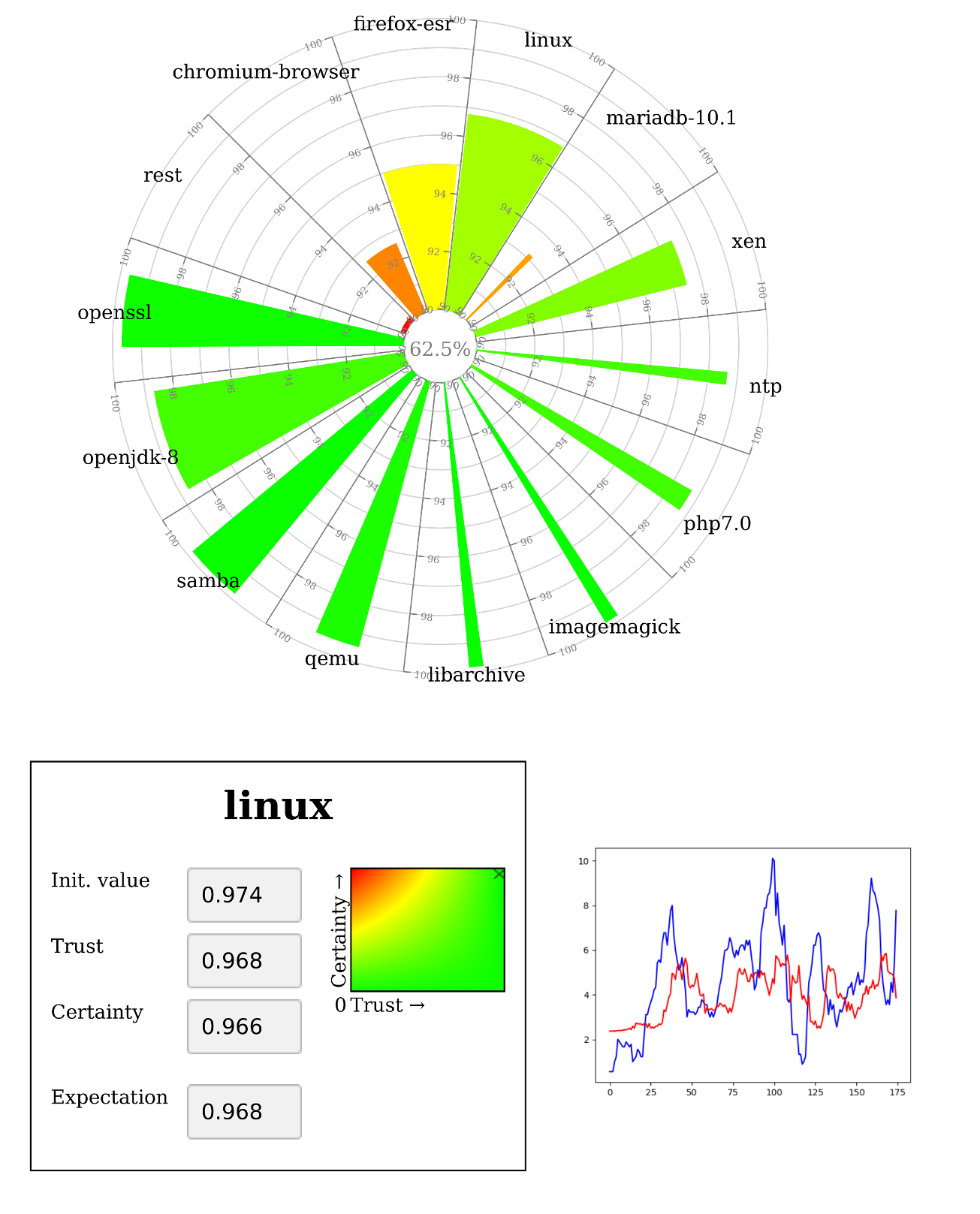}
\caption{Trustworthiness assessment of the full-fledged system.}
\label{fig:example}
\end{figure}

%
\section{Related Work}\label{sec:related}
We first present some significant recent work in the area of vulnerability discovery,
which is adjacent and complimentary to our work, and then discuss related work
in the field of predicting vulnerable software components. Finally, we go over
work in the field of software risk and trustworthiness assessment.
\medskip
\par\noindent\textbf{Vulnerabilities and malware discovery.}
There is a lot of ongoing work in the field of automatic vulnerability discovery
in software. In \cite{yamaguchi2013chucky}, the authors combine techniques from
static analysis and machine learning to identify missing checks that lead to vulnerabilities
in several software projects, like the Linux kernel and Pidgin.
In \cite{yamaguchi2014modeling}, a novel representation of code, namely code property graphs,
are introduced and their usefulness in identifying software vulnerabilities is showcased,
while \cite{yamaguchi2015automatic} deals with the automated finding of taint-style
vulnerabilities.
Other approaches that make use of virtualization and other forms of dynamic analysis
techniques have been recently proposed, e.g. in Digtool~\cite{pan2017digtool}, which finds vulnerabilities
at the binary level, in \cite{petsios2017nezha} where differential testing is used, in~\cite{stephens2017driller}
where the authors use fuzzing in combination with symbolic execution, and in~\cite{schumilo2017kafl}
where hardware-assisted fuzzing is established. Regarding automated discovery and classification of
malware, \cite{arp2014drebin} offers a lightweight tool for detection of malware in Android.
Our work is orthogonal, but complimentary to the contributions highlighted above, in the sense
that they handle specific cases of vulnerabilities and as seen in the real world, even when
related discovery tools are employed, end products still contain a multitude of vulnerabilities.
Techniques similar to the above can also act as evidence sources for M-STAR's trustworthiness
calculation.
\medskip
\par\noindent\textbf{Vulnerabilities prediction.}
The area of predicting which software components are more likely to contain
vulnerabilities has also yielded some prominent results. The pioneering work of
Neuhaus et al.~\cite{neuhaus2007predicting} analyzed C/C++ files of the Mozilla
codebase and classified them as vulnerable or not using SVMs. Specifically, components
that had similar imports and function calls were shown likely to share vulnerability status.
In~\cite{neuhaus2009beauty}, the authors leverage dependency relationships to classify Red Hat
linux packages, whereas in \cite{scandariato2014predicting}, text mining of the source code is employed to predict if a given
component is vulnerable. Finally, in~\cite{roumani2015time}, linear autoregressive models are shown to be
reasonably accurate at forecasting vulnerabilities, and in~\cite{jimenez2016vulnerability}, a comparison of proposed
prediction techniques is performed on Linux kernel components.
\medskip
\par\noindent\textbf{Software trustworthiness and risk.}
Research regarding the trustworthiness of software, especially w.r.t. its security
properties goes back to the \emph{Trusted Software Methodology} \cite{amoroso1994process},
a process-oriented methodology developed in the 90's. In a recent report~\cite{boland2010toward}, NIST
proposes a framework for assessing software trustworthiness by weighting in evaluations carried
out either by automatic code checkers or by experts, in order to deduce an overall trustworthiness assessment
for the component under question. Our system on the contrary does not rely on expert opinions, although our model
can readily accommodate them. Our work is most closely related to that of Bugiel et al.~\cite{bugiel2011scalable},
where the authors propose a tool for software trustworthiness assessment of Debian systems. Although this work
served as an inspiration to our system, our work differs considerably in nature,
mainly because it provides a mathematically and empirically verified solution to trustworthiness assessment,
contrary to the ad-hoc approach of the aforementioned paper.

\section{Conclusion, Limitations and future work}\label{sec:concl}
In this paper we presented M-STAR, a complete framework for assessing the
trustworthiness of software systems. M-STAR's modular design offers adaptability to
various use-cases and different technologies. We employ state-of-the-art
prediction techniques and Bayesian probabilistic models, known
as computational trust models, in order to model and calculate our trustworthiness assessments.
Our prototype, written in python, will be made publicly available as a web interface,
and the code will be published on github.
During our detailed investigation of vulnerabilities in the Debian ecosystem,
we came to the conclusion that proactive techniques like M-STAR are necessary,
as there is no observable decrease in the vulnerability rate of software in Debian,
even when considering single versions of software.

We believe that M-STAR is an important contribution
towards assessing the real-world security of systems.
However, M-STAR is only a first step
towards this goal and consequently exhibits some limitations, which should be addressed in
future work. Namely, experimenting with other datasets (other than Debian) and
different prediction techniques, would be highly beneficial, as static and dynamic analysis
tools become more generic and accurate. Furthermore, deploying M-STAR in the wild and
performing user studies among system administrators regarding the impact of M-STAR's assessments
on their choices, would also be interesting.

\section*{Acknowledgements}
This work has been co-funded by the DFG as part of project S1 within the CRC 1119
CROSSING.
\begin{acronym}
	\acro{MLE}{Maximum Likelihood Estimator}
	\acro{DSA}{Debian Security Advisory}
	\acro{SVM}{Support Vector Machine}
	\acro{LSTM}{Long Short-Term Memory}
	\acro{SMC}{Secure Multiparty Computation}
\end{acronym}



\bibliographystyle{IEEEtran}
\bibliography{IEEEabrv,../bib/paper}
\appendices
\section{CertainLogic standard (logical) operators}\label{ap:clogic_standard}

\medskip
\par\noindent\textbf{CertainLogic $\boldsymbol{AND^{ct}}$ ($\boldsymbol{\land^{ct}}$) Operator:} The operator $\wedge^{ct}$ is applicable when opinions about two independent
propositions need to be combined in order to produce a new opinion reflecting the degree of truth of both propositions
simultaneously. Note that the opinions are represented using the CertainTrust model.
The rationale behind the definitions of the logical operators of CertainTrust (e.g., $AND^{ct}$ ($\land^{ct}$)) demands an analytical discussion. 

In standard binary logic, logical operators operate on propositions that only consider the values `TRUE' or `FALSE'
(i.e., $1$ or $0$ respectively) as input arguments. In standard probabilistic logic, the logical operators operates on propositions
that consider values in the range of $[0,1]$ (i.e., probabilities) as input arguments. However, logical operators in the standard
probabilistic approach are not able to consider \emph{uncertainty} about the probability values. Subjective Logic's~\cite{josang2001logic} logical operators
are able to operate on opinions that consider uncertain probabilities as input arguments. Additionally, Subjective Logic's logical
operators are a generalized version of standard logic operators and probabilistic logic operators.

CertainLogic's logical operators operate on CertainTrust's opinions, which represent uncertain probabilities in a more flexible and simple
manner than the opinion representation in Subjective Logic (SL). Note that CertainTrust's representation and Subjective Logic's
representation of opinions are \emph{isomorphic} with the mapping provided in~\cite{Ries2009Trust}. For a detailed discussion on the
representational model of Subjective Logic's opinions and CertainTrust's opinions, we refer the readers to Chapter 2 of~\cite{HabibThesis2014}.
The definitions of CertainLogic's logical operators are formulated in a way so that they are equivalent to the definitions of logical operators
in Subjective Logic. This equivalence serves as an argument for the \emph{justification} and \emph{mathematical validity} of CertainLogic
logical operators' definitions. Moreover, these operators are a generalization of binary logic and probabilistic logic operators. 




\begin{definition}[Operator $AND^{ct}$]\label{Defn:CT_AND}
	Let $A$ and $B$ be two independent propositions and the opinions about the truth of these propositions be given as $o_A=(t_A,c_A,f_A)$ and $o_B=(t_B,c_B,f_B)$, respectively. Then, the resulting opinion is denoted as $o_{A \land^{ct} B}=(t_{A \land^{ct} B},c_{A \land^{ct} B},f_{A \land^{ct} B})$ where $t_{A \land^{ct} B}$, $c_{A \land^{ct} B}$, and $f_{A \land^{ct} B}$ are defined in Table~\ref{tab:operators} ($AND$). We use the symbol '$\land^{ct}$' to designate the operator $AND^{ct}$ and we define $o_{A \land^{ct} B} \equiv o_A \land^{ct} o_B$.
\end{definition}

The aggregation (using the $AND^{ct}$ operator) of opinions about independent propositions $A$ and $B$ are formulated in a way that the resulting
initial expectation ($f$) is dependent on the initial expectation values, $f_{A}$ and $f_{B}$ assigned to $A$ and $B$ respectively.
Following the equivalent definitions of Subjective Logic's normal conjunction operator and basic characteristics of the same operator ($\land$)
in standard probabilistic logic, we define $f_{A \land^{ct} B}=f_A f_B$. The definitions for $c_{A \land^{ct} B}$ and $t_{A \land^{ct} B}$ are formulated in
similar manner and the corresponding adjustments in the definitions are made to maintain the equivalence between the operators of Subjective Logic
and CertainLogic. The $AND^{ct}$ ($\land^{ct}$) operator of CertainLogic is associative and commutative; both properties are desirable for the evaluation
of propositional logic terms (PLTs).

\begin{table*}[!tb]
	\caption{Definition of the operator $AND^{ct}$ ($\land^{ct}$)}
\label{tab:operators}
\scriptsize
\centering
\begin{tabular}
{|c|}
\hline
\\
$AND^{ct}$\tabularnewline 
\hline
\\
$
\begin{aligned}

	c_{A \land^{ct} B}=&
\begin{cases}
c_A+ c_B-c_A c_B-\frac{\left(1-c_A\right) c_B  \left(1-f_A\right) t_{B}+c_A  \left(1-c_{B}\right) \left(1-f_B\right) t_A}{1-f_A f_B}&\text{if }f_A f_B \neq 1,\\
\text{``undefined''}~\text{else}.\\
\end{cases}\\
t_{A \land^{ct} B}=&
\begin{cases}
	\frac{1}{c_{A \land^{ct} B}} \left(c_A c_B t_A t_B+\frac{c_A(1-c_B) (1-f_A)f_B t_A+ (1-c_A)c_B f_A(1-f_B)t_B}{1-f_A f_B}\right)~\text{if }c_{A \land^{ct} B} \neq 0 \text{ and } f_A f_B \neq 1,\\
        0.5~\text{else}.\\
\end{cases}\\
f_{A \land^{ct} B}=&f_A f_B\\
\end{aligned}
$
\tabularnewline
\\
\hline
\end{tabular}
\end{table*}

\medskip
\par\noindent\textbf{CertainLogic $\boldsymbol{OR^{ct}}$ ($\boldsymbol{\lor^{ct}}$) Operator:}
The operator $\lor^{ct}$ is applicable when opinions about two independent propositions need to form a new opinion reflecting the degree
of truth for at least one out of two propositions. 

\begin{definition}[Operator $OR^{ct}$]\label{Defn:CT_OR}~
Let $A$ and $B$ be two independent propositions and the opinions about the truth of these
propositions be given as $o_A=(t_A,c_A,f_A)$ and $o_B=(t_B,c_B,f_B)$, respectively. Then,
the resulting opinion is denoted as $o_{A \lor^{ct} B}=(t_{A \lor^{ct} B},c_{A \lor^{ct} B},f_{A \lor^{ct} B})$
where $t_{A \lor^{ct} B}$, $c_{A \lor^{ct} B}$, and $f_{A \lor^{ct} B}$ are defined in Table~\ref{tab:operators}
($OR^{ct}$). We use the symbol '$\lor^{ct}$' to designate the operator $OR^{ct}$ and we define
$o_{A \lor^{ct} B} \equiv o_A \lor^{ct} o_B$.
\end{definition}

The aggregation (using the $OR^{ct}$ operator) of opinions about independent propositions $A$ and $B$
is formulated in a way that the resulting initial expectation ($f$) is dependent on the initial
expectation values, $f_{A}$ and $f_{B}$ assigned to $A$ and $B$ respectively. Following the equivalent
definitions of Subjective Logic's normal disjunction operator and the basic characteristics of the
same operator ($\lor$) in standard probabilistic logic, we define $f_{A \lor^{ct} B}=f_A+f_B-f_A f_B$.
The definitions for $c_{A \lor^{ct} B}$ and $t_{A \lor^{ct} B}$ are formulated in similar manner and the
corresponding adjustments in the definitions are made to maintain the equivalence between the operators
of Subjective Logic and CertainLogic. The $OR^{ct}$ ($\lor^{ct}$) operator of CertainLogic is associative and commutative; both properties are desirable for the evaluation
of PLTs.

\begin{table*}[!tb]
	\caption{Definition of the operator $OR^{ct}$ ($\lor^{ct}$)}
\label{tab:operators2}
\scriptsize
\centering
\begin{tabular}
{|c|}
\hline
\\
$OR^{ct}$\tabularnewline 
\hline\\
$
\begin{aligned}

	c_{A \lor^{ct} B}=&
\begin{cases}
c_A+c_B-c_A c_B-\frac{c_A (1-c_B)f_B(1- t_A)+(1-c_A) c_B f_A (1-t_B)}{f_A+f_B-f_A f_B}&\text{if }f_A f_B \neq 0,\\
\text{``undefined''}~\text{else}.\\
\end{cases}\\
t_{A \lor^{ct} B}=&
\begin{cases} 
	\frac{1}{c_{A \lor^{ct} B}}\left(c_A t_A + c_B t_B - c_A c_B t_A t_B \right)~\text{if }c_{A \lor^{ct} B} \neq 0,\\
        0.5~\text{else}.\\
\end{cases}\\
f_{A \lor^{ct} B}=&f_A+f_B-f_A f_B\\
\end{aligned}
$
\tabularnewline
\\
\hline
\end{tabular}
\end{table*}

\section{CertainLogic non-standard (FUSION) operators}
\label{ap:fusion}
Assume that one wants to fuse conflicting opinions (about a proposition) derived from multiple sources. In this case, one should use the conflict-aware fusion ($C.FUSION$) operator as defined in~\cite{HabibThesis2014}. This operator operates on dependent conflicting opinions and reflects the calculated degree of conflict ($DoC$) in the resulting fused opinion. Note that the $C.FUSION$ operator is also able to deal with preferential weights associated with opinions. 

\begin{definition}[$C.FUSION$]\label{Defn:CT_Fusion_conflict_n_ary}

Let $A$ be a proposition and let $o_{A_1} = (t_{A_1},c_{A_1},f_{A_1})$, $o_{A_2} = (t_{A_2},c_{A_2},f_{A_2})$,$\dotsm$, $o_{A_n} = (t_{A_n},c_{A_n},f_{A_n}$)  be $n$ opinions associated to $A$. Furthermore, the weights $w_1$, $w_2$,$\dotsm$, $w_n$ (with $w_1, w_2,\dotsm, w_n \in \mathbb{R}_{0}^{+}$ and $w_1+ w_2+\dotsm+w_n\neq0$) are assigned to the opinions $o_{A_1}$, $o_{A_2}$,$\dotsm$, $o_{A_n}$, respectively.
\smallskip The \textbf{conflict-aware fusion} is  denoted as $$o_{ \FUSIONC (A_{1},\dotsm, A_{n})}=$$ $$=((t_{\FUSIONC (A_{1},\dotsm, A_{n})},c_{\FUSIONC (A_{1},\dotsm, A_{n})},f_{\FUSIONC (A_{1},\dotsm, A_{n})}),DoC )$$ where $t_{\FUSIONC (A_{1},\dotsm, A_{n})}$, $c_{\FUSIONC (A_{1},\dotsm, A_{n})}$, $f_{\FUSIONC (A_{1},\dotsm, A_{n})}$,\smallskip\\ and the degree of conflict $DoC$ are defined in Table~\ref{tab:operators_conflict}. We use the symbol ($\FUSIONC$) to designate the operator $C.FUSION$ and we define: $$o_{\FUSIONC (A_{1},\dotsm, A_{n})} \equiv \FUSIONC ((o_{A_1},w_1),(o_{A_2},w_2),\dotsm, (o_{A_n},w_n))$$
\end{definition}

\begin{table*}[!tb]
\caption{Definition of the Conflict-aware Fusion Operator}
\label{tab:operators_conflict}
\scriptsize
\centering
\begin{tabular}
{|c|}
\hline
$
\begin{aligned}
t_{\FUSIONC (A_{1}, A_{2},\dotsm, A_{n})}=&
\begin{cases}
                \frac{\displaystyle\sum^n_{i=1} w_i t_{A_i}}{\displaystyle\sum^n_{i=1} w_i}~\text{if }c_{A_1} = c_{A_2} =\dotsm=c_{A_n}= 1 \enspace,\\
    0.5~\text{if }c_{A_1} = c_{A_2} =\dotsm=c_{A_n}= 0 \enspace,\\
                \frac{\displaystyle\sum^n_{i=1}(c_{A_i} t_{A_i} w_i \displaystyle\prod^{\substack{n}}_{\substack{j=1,~j\neq i}} (1-c_{A_j}))}{\displaystyle\sum^n_{i=1} (c_{A_i} w_i \displaystyle\prod^{\substack{n}}_{\substack{j=1,~j\neq i}} (1-c_{A_j}))}~\text{if }\{c_{A_i}, c_{A_j}\}\neq 1 \enspace.\\
\end{cases}
\\
\\

c_{\FUSIONC (A_{1}, A_{2},\dotsm, A_{n})}=&
\begin{cases}
                        1*(1-DoC)~\text{if }c_{A_1} = c_{A_2}=\dotsm=c_{A_n}= 1 \enspace,\\
      \frac{\displaystyle\sum^n_{i=1}(c_{A_i} w_i \displaystyle\prod^{\substack{n}}_{\substack{j=i+1}} (1-c_{A_j}))}{\displaystyle\sum^n_{i=1} (w_i \displaystyle\prod^{\substack{n}}_{\substack{j=1,~j\neq i}} (1-c_{A_j}))}*(1-DoC)~\text{if }\{c_{A_i}, c_{A_j}\}\neq 1  \enspace.\\
\end{cases}
\\
\\
f_{\FUSIONC (A_{1}, A_{2},\dotsm, A_{n})}=&\frac{\displaystyle\sum^n_{i=1} w_i f_{A_i}}{\displaystyle\sum^n_{i=1} w_i}
\\
\\
DoC=&\frac{{\displaystyle\sum^{\substack{n}}_{\substack{i=1}} \displaystyle\sum^{\substack{n}}_{\substack{j=1,j\neq i}} DoC_{A_i, A_j}}}{\frac{n(n-1)}{2}}
\\
\\
DoC_{A_i, A_j}=&\left\lvert{t_{A_i}-t_{A_j}}\right\rvert*c_{A_i}*c_{A_j}*\left(1-\left\lvert\frac{w_i-w_j}{w_i+w_j}\right\rvert\right)\\
\end{aligned}
$
\tabularnewline
\hline
\end{tabular}
\end{table*}

The conflict-aware fusion ($C.FUSION$) operator is commutative and idempotent, but not associative.

The rationale behind the definition of the \emph{conflict-aware} fusion demands an extensive discussion. The basic concept of this operator is that the operator extends CertainLogic's \emph{Weighted fusion}~\cite{HabibFusionTrust2012} operator by calculating the degree of conflict ($DoC$) between a pair of opinions. Then, the value of $(1-DoC)$ is multiplied with the certainty ($c$) that would be calculated by the weighted fusion (the parameters for $t$ and $f$ are the same as in the weighted fusion).

Now, we discuss the calculation of the $DoC$ for two opinions. For the parameter, it holds $DoC \in [0,1]$. This parameter depends on the trust value ($t$), the certainty values ($c$), and the weights ($w$). The weights are assumed to be selected by the trustors (consumers)  and the purpose of the weights is to model the preferences of the trustor when aggregating opinions from different sources. We assume that the compliance of their preferences are ensured under a policy negotiation phase. For example, users might be given three choices, High ($2$), Low ($1$) and No preference ($0$, i.e., opinion from a particular source is not considered), to express their preferences on selecting the sources that provide the opinions. Note that the weights are not introduced to model the reliability of sources. In this case, it would be appropriate to use the discounting operator~\cite{Ries2009Trust,josang2001logic} to explicitly consider reliability of sources and apply the fusion operator on the results to influence users' preferences. The values of $DoC$ can be interpreted as follows: 
\begin{itemize}
\item \textbf{No conflict ($DoC=0$):} For $DoC=0$, it holds that there is \emph{no conflict} between the two opinions. This is true if both opinions agree on the trust value, i.e., $t_{A_1}=t_{A_2}$ or in case that at least one opinion has a certainty $c=0$ (for completeness we have to state that it is also true if one of the weights is equal to $0$, which means the opinion is not considered).
\item \textbf{Total conflict ($DoC=1$):} For $DoC=1$, it holds that the two opinions are weighted equally ($w_1=w_2$) and contradicts each other to a maximum. This means, that both opinions have a maximum certainty ($c_{A_1}=c_{A_2}=1$) and maximum divergence in the trust values, i.e., $t_ {A_1}=0$ and $t_ {A_2}=1$ (or $t_ {A_1}=1$ and $t_ {A_2}=0$).
\item \textbf{Conflict ($DoC\in]0,1[$):} For $DoC\in]0,1[$, it holds that there are two opinions contradict each other to a certain degree. This means that the both opinions does not agree on the trust values, i.e., $t_{A_1}\neq t_{A_2}$, having certainty values other than $0$ and $1$. The weights can be any real number other than $0$.
\end{itemize}

Next, we argue for integrating the degree of conflict, $DoC$, into the resulting opinion by multiplying the certainty with $(1-DoC)$. The argument is, in case that there are two (equally weighted) conflicting opinions, then this indicates that the information which these opinions are based on is not representative for the outcome of the assessment or experiment. Thus, for the sake of representativeness, in the case of total conflict (i.e., $DoC=1$), we reduce the certainty ($c_{(o_{A_1},w_1) \FUSION (o_{A_2},w_2)}$) of the resulting opinion by a multiplicative factor, $(1-DoC)$. The certainty value is $0$ in this case. 

For $n$ opinions, degree of conflict (i.e., $DoC_{A_i,A_j}$) in Table~\ref{tab:operators_conflict} is calculated for each opinion pairs. For instance, if there are $n$ opinions there can be at most $\frac{n(n-1)}{2}$ pairs and degree of conflict is calculated for each of those pairs individually. Then, all the pair-wise $DoC$ values are averaged, i.e., averaging $\frac{n(n-1)}{2}$ pairs of $DoC_{A_i,A_j}$. Finally, the certainty (i.e., $c_{\FUSIONC (A_{1}, A_{2},\dotsm, A_{n})}$) parameter of the resulting opinion (see Table~\ref{tab:operators_conflict}) is adjusted with the resulting $DoC$ value.

In Table~\ref{tab:operators_conflict}, for all opinions if it holds $c_{A_i}=0$ (complete uncertainty), the expectation values depends only on $f$. However, for soundness we define $t_{A_i}=0.5$ in this case.

\section{Deployment scenarios}
After describing M-STAR in the previous sections, we proceed with
proposing and discussing two possible deployment scenarios for our
system.
\medskip
\par\noindent\textbf{Software installation process:}
Our first scenario is using M-STAR's assessments to assist a system administrator understand
and minimize the security risks of their system configuration. For example, more risky components
can be substituted by more trustworthy ones with equivalent functionality. Figure~\ref{fig:example}
shows the trustworthiness assessment of a Debian system, the components of which are considered
equally critical to the security of the system. This is due to the fact that vulnerabilities of
any of the components can lead to an adversary taking control of the system.

To fully exploit the rich expressiveness of M-STAR's trust modeling and the logical operators
in our disposal (see Section~\ref{sec:combining}), we can apply M-STAR to the problem of risk assessment
and minimization, in the context of \ac{SMC}. \ac{SMC} was initially introduced by Yao~\cite{yao1982protocols}
for the two-party setting and soon extended to the multi-party setting. The basic idea of \ac{SMC} in the client-server
model, is that $m$ servers can run an algorithm with input data provided by the clients in secret-shared
format and generate an output without learning anything about the input data of the clients, assuming $k-out-of-m$
of the servers are honest. Although \ac{SMC} remained a theoretical construct for many years, recent advances have
allowed the realization of the concept, as seen e.g. in the use-case of anonymous messaging~\cite{alexopoulos2017mcmix}.
One of the main practical issues of \ac{SMC} is that the software configurations of the servers is critical to
guaranteeing that a vulnerability of a software component cannot affect more than $k$ servers, thus breaking the
security of the whole system. As we showcased in Section~\ref{sec:combining} for the case of $1-out-of-2$ servers,
systems administrators can use M-STAR to maximize the overall trustworthiness of multi-server systems by
adapting their software configurations.
\medskip
\par\noindent\textbf{Trust-based access control:}
Traditional access control models are based solely on cryptographic credentials. However, assessing only the
identity of the party requesting access to potentially valuable information overlooks the possibility that
the device used by the otherwise honest party is not trustworthy. Consequently, using trust and risk in
access control policies was proposed, e.g. in~\cite{cheng2007fuzzy} and more recently in the case of Intel in~\cite{evered2013android}.
M-STAR trust scores can be readily used in access control models of this kind, providing a well-founded and probabilistic
measure of trustworthiness w.r.t. the software configuration of the party requesting access.

\section{Additional Figures}
\begin{figure}[h]
\centering
\includegraphics[angle=270, scale=0.53]{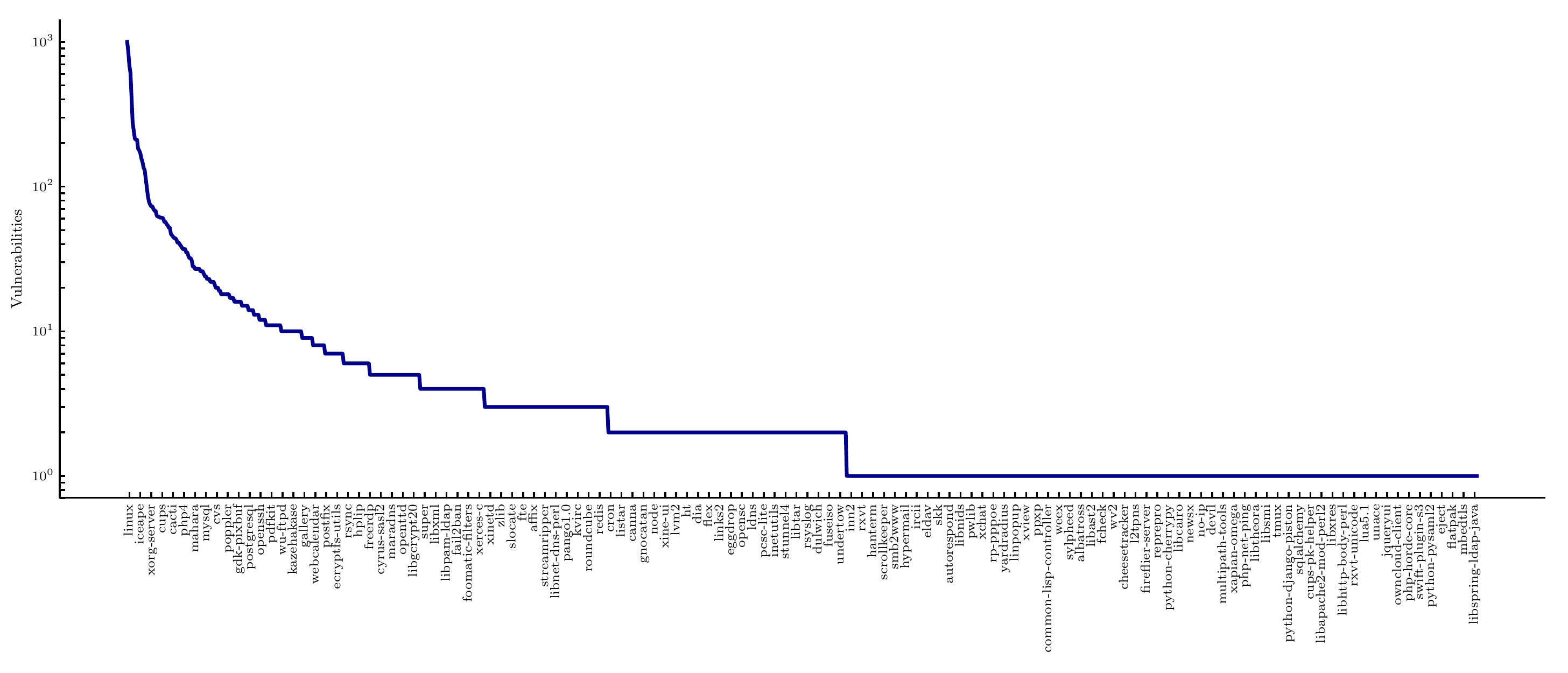}
\caption{The distribution of vulnerabilities in the Debian ecosystem (years 2001-2016). The scale of axis y is logarithmic. All packages are taken into account. Every tenth package name appears on the x axis for space reasons.}
\label{fig:distr}
\end{figure}

\end{document}